\RequirePackage{fix-cm}
\RequirePackage{fixltx2e}
\documentclass[runningheads,algo2e,ruled,vlined]{extarticle}
\usepackage[utf8]{inputenc}
\usepackage{geometry}
\usepackage{color}
\usepackage{verbatim}
\usepackage{float}
\usepackage{algorithm2e}
\usepackage{amsmath}
\usepackage{amsthm}
\usepackage{amssymb}
\usepackage{graphicx}
\usepackage[unicode=true,pdfusetitle,
 bookmarks=true,bookmarksnumbered=true,bookmarksopen=true,bookmarksopenlevel=1,
 breaklinks=true,pdfborder={0 0 0},pdfborderstyle={},backref=false,colorlinks=true]
 {hyperref}
\hypersetup{
 linkcolor=blue, citecolor=blue, urlcolor=blue, filecolor=blue,pdfpagelayout=OneColumn, pdfnewwindow=true,pdfstartview=XYZ, plainpages=false, pdfpagelabels, hyperindex=true}

\makeatletter

\newtheorem{myDef}{Definition}[]

\providecommand{\tabularnewline}{\\}

\theoremstyle{plain}
\newtheorem{thm}{\protect\theoremname}
\theoremstyle{plain}
\newtheorem{assumption}{\protect\assumptionname}
\theoremstyle{plain}
\newtheorem{cor}{\protect\corollaryname}

\usepackage{booktabs}
\usepackage{caption}
\usepackage{subcaption}
\usepackage{pgfplots}
\usepackage[all]{nowidow}
\usepackage{tikz}
\usetikzlibrary{er,positioning,bayesnet}
\usepackage{multicol}
\usepackage[inline]{enumitem}
\usepackage{natbib}
\usepackage{times}
\usepackage{fancyhdr}
\usepackage{algpseudocode}
\usepackage{algorithmicx}\usepackage{xcolor}
\definecolor{algoColorKeyword}{named}{blue}
\definecolor{algoColorComment}{named}{olive}

\SetKwComment{Comment}{$\triangleright$\ }{}
\renewcommand{\subsubsection}{\@startsection{subsubsection}{3}{\z@}%
{-3.25ex\@plus -1ex \@minus -.2ex}%
{-1.5ex \@plus -.2ex}%
{\normalfont\normalsize\bfseries}}

\include{pythonlisting}

\definecolor{blue}{HTML}{1F77B4}
\definecolor{orange}{HTML}{FF7F0E}
\definecolor{green}{HTML}{2CA02C}

\pgfplotsset{compat=1.14}

\makeatother

\providecommand{\assumptionname}{Assumption}
\providecommand{\corollaryname}{Corollary}
\providecommand{\theoremname}{Theorem}

\begin{document}
\title{Markovian approximation of the rough Bergomi model for Monte Carlo option pricing}
\author{Qinwen Zhu\thanks{School of Mathematical Sciences, Nanjing Normal University, Nanjing, 210023, PR China. Email: \href{mailto:qinwen.wendy.zhu@gmail.com}{qinwen.wendy.zhu@gmail.com}}
,\and Gr{\'e}goire Loeper\thanks{School of Mathematics, Monash University, Clayton VIC Australia}\hspace{0.333em}\thanks{Centre for Quantitative Finance and Investment Strategies, Monash University, Clayton VIC Australia}
,\and Wen Chen\thanks{CSIRO Data61, RiskLab, Docklands VIC Australia}
,\and Nicolas Langren{\'e}\footnotemark[4]}
\maketitle
\begin{abstract}

\noindent The recently developed rough Bergomi (rBergomi) model is
a rough fractional stochastic volatility (RFSV) model which can generate
more realistic term structure of at-the-money volatility skews compared
with other RFSV models. However, its non-Markovianity brings mathematical
and computational challenges for model calibration and simulation.
To overcome these difficulties, we show that the rBergomi model can
be approximated by the Bergomi model, which has the Markovian property.
Our main theoretical result is to establish and describe the affine
structure of the rBergomi model. We demonstrate the efficiency and
accuracy of our method by implementing a Markovian approximation algorithm
based on a hybrid scheme.
\noindent \vspace{0.25em}


\noindent \vspace{0.25em}

\textbf{ACM codes}: G.3; I.6.1; F.2.1; G.1.2; I.6.3; G.1.10;\textbf{ }
\end{abstract}

\section{Introduction}

The rough Bergomi (rBergomi) model introduced by \citet{bayer2016pricing}
has gained acceptance for stochastic volatility modelling due to its
power-law at-the-money volatility skew which is consistent with empirical
studies (see \citealt{forde2017asymptotics}, \citealt{fukasawa2017short},
\citealt{gatheral2018volatility}) and the market impact function
under the no-arbitrage assumption (see \citealt{jusselin2018no}).
However, the stochastic process which characterizes this volatility
model is rougher than that of a Brownian motion; in particular, the
lack of Markovianity makes classical pricing methods infeasible.

In order to price options under an rBergomi model and calibrate such
a model, \citet{bayer2018hierarchical} propose hierarchical adaptive
sparse grids for option pricing, \citet{bayer2019deep} propose a
deep learning method for rBergomi model calibration, \citet{jacquier2018vix}
develop pricing algorithms for VIX futures and options, and \citet{mccrickerd2018turbocharging}
develop a `turbocharged' Monte Carlo pricing method. In spite of these
efforts, the inherent challenges brought by the rBergomi model still
prevent its widespread adoption in industry.

Inspired by the technique from \citet{abi2019multifactor}, \citet{gatheral2019affine}
and \citet{harms2019affine}, in which the authors design a multi-factor
stochastic volatility model with Markovian structure to approximate
the rough Heston model, we establish an analogous multi-factor affine
structure for the rBergomi model. In the affine structure, the Volterra
kernel corresponds to a superposition of infinitely many Ornstein-Uhlenbeck
(O-U) processes with different speeds of mean reversion. Truncating
this infinite sum into a finite sum of O-U processes yields an approximation
of the rBergomi model which is a Markovian approximated Bergomi (aBergomi)
model. We then prove the existence and uniqueness of solutions to this affine aBergomi
model, and show that its affine structure converges to the one of
the rBergomi model.

To numerically simulate the rBergomi model in practice, we adopt the
hybrid scheme proposed in \citet{bennedsen2017hybrid} for the stochastic
Volterra-type integrals $\tilde{X}=\sqrt{2\alpha+1}\int_{s}^{t}(t-s)^{\alpha}dW_{s}$
($\kappa=1$). The hybrid scheme consists in approximating the power-law
kernel $K_{\text{pow}}=\sqrt{2\alpha+1}(t-s)^{\alpha}$ by a combination
of a power function near zero and a step function elsewhere, with
a lower $\mathcal{O}(N\log N)$ complexity, where $N$ is the number
of time steps, as opposed to the $\mathcal{O}(N^{3})$ complexity
of the Cholesky method in \citet{bayer2016pricing}. Here, the rBergomi
power-law kernel $K_{\text{pow}}$ can be approximated by the exponential
kernel $K_{\text{exp}}=\sum_{i=1}^{n}\alpha_{i}e^{-\kappa_{i}(t-s)}$
after truncating $K_{\text{pow}}$ before $t$. Our numerical tests
demonstrate that using $n=25$ exponential terms in $K_{\text{exp}}$
(i.e. a 25-term O-U process), we can obtain an accurate (low root
mean squared error), yet tractable and computationally efficient approximation
of the fractional rBergomi model.

The paper is organized as follows. In Section \ref{sec:rBergomi},
we introduce the Bergomi and rBergomi models and discuss their respective
ATM volatility skews. In particular we provide for the first time
a proof that the ATM volatility skew of the rBergomi model is equivalent
to the power $T^{H-\frac{1}{2}}$ while this does not hold for the
Bergomi model (equation \eqref{eq:2Dvolskew}). In Section \ref{sec:Markovian},
we establish the affine structure of the rough Bergomi model. Section
\ref{sec:rBergomi_approximation} is dedicated to the approximation
of the rough Bergomi model by a multi-factor Bergomi model. Finally,
Section \ref{sec:numerical_method} compares numerical simulations
of the rBergomi model with our approximated Bergomi (aBergomi) model
with a finite number of terms, showing the effectiveness of our approximation.
\section{Bergomi and rough Bergomi models\label{sec:rBergomi}}

This section introduces the Bergomi and rough Bergomi stochastic volatility
models (Definitions \ref{def:Bergomi} and \ref{def:rBergomi}), along
with the corresponding notations used throughout the paper.

We consider a filtered probability space $(\Omega,\mathcal{F},(\mathcal{F}_{t})_{t\in\mathbb{R}},\mathbb{Q})$,
which supports two dimensional correlated Brownian motions $W$ and $B$. A log
price process $X_{t}:=\log(S_{t})$ is assumed to follow the dynamics
\begin{equation}
dX_{t}=-\frac{1}{2}V_{t}dt+\sqrt{V_{t}}dW_{t}\,,\label{eq:log_price_dynamics}
\end{equation}
where $V_{t}\geq0$ is the instantaneous spot variance process. Let
$\xi_{t}^{u},u\geq t$ be the instantaneous forward variance for date
$u$ observed at time $t$; in particular $\xi_{t}^{t}=V_{t}$ corresponds
to the spot variance.

\citet{bayer2016pricing} proposed the so-called \textit{rough Bergomi
model} where the forward variance follows
\begin{equation}
d\xi_{t}^{u}=\xi_{t}^{u}\eta\sqrt{2\alpha+1}\left(u-t\right){}^{\alpha}dB_{t},~u\geq t\label{eq:forward_variance_rBergomi}
\end{equation}
where $W$ and $B$ have correlation $\rho$, $\alpha\triangleq H-\frac{1}{2}\in\left(-\frac{1}{2},0\right)$
is a negative exponent depending on the Hurst exponent $H\in\left(0,\frac{1}{2}\right)$
of the underlying fractional Brownian motion, and $\eta$ is a positive
parameter depending on $H$. The definition of the rBergomi model
is summarized below:

\begin{myDef}
\label{def:rBergomi}The rBergomi stochastic volatility model takes
the form
\begin{equation}
\left\{ \begin{aligned} & dX_{t}=-\frac{1}{2}V_{t}dt+\sqrt{V_{t}}dW_{t},\\
 & d\xi_{t}^{u}=\xi_{t}^{u}\eta\sqrt{2\alpha+1}(u-t)^{\alpha}dB_{t},
\end{aligned}
\right.\label{eq:rBergomi}
\end{equation}
where $\alpha=H-\frac{1}{2}\in\left(-\frac{1}{2},0\right)$, and $d\left\langle W,B\right\rangle _{t}=\rho dt$.
\end{myDef}
By contrast, the two-factor Bergomi model is defined as follows.

\begin{myDef}
\label{def:Bergomi}The two-factor Bergomi model (\citealt{bergomi2005smile},
\citealt{bergomi2009smile}) is defined by:
\begin{equation}
\left\{ \begin{aligned} & dX_{t}=-\frac{1}{2}V_{t}dt+\sqrt{V_{t}}dW_{t}^{S},\\
 & d\xi_{t}^{u}=\xi_{t}^{u}\alpha_{\theta}\omega\left(\left(1-\theta\right)e^{-\kappa_{X}(u-t)}dW_{t}^{X}+\theta e^{-\kappa_{Y}(u-t)}dW_{t}^{Y}\right)
\end{aligned}
\right.\label{eq:Bergomi}
\end{equation}
with
\[
\begin{aligned} & d\langle W^{S},W^{X}\rangle_{t}=\rho_{SX}dt\\
 & d\langle W^{S},W^{Y}\rangle_{t}=\rho_{SY}dt\\
 & d\langle W^{X},W^{Y}\rangle_{t}=\rho_{XY}dt,
\end{aligned}
\]
where $\xi_{t}^{t}=V_{t}=\omega$ is the lognormal volatility of the
instantaneous variance under the normalizing factor $\alpha_{\theta}=\left(\left(1-\theta\right){}^{2}+2\rho_{XY}\theta\left(1-\theta\right)+\theta^{2}\right){}^{-\frac{1}{2}}$
and $\theta$ is a mixing parameter of the short term factor driven
by $W^{X}$ and the long term factor driven by $W^{Y}$ ($\kappa_{X}>\kappa_{Y}$).
\end{myDef}

\begin{assumption}
\label{assu:initial_curve}Without loss of generality, we assume throughout the
paper that the initial forward variance curve $\xi_{0}^{u},u\geq0$
is flat. This simplification is common in the rBergomi literature,
see for example \citet{bayer2016pricing}, \citet{bayer2018hierarchical}
and \citet{bayer2019deep}. We use the notation $\xi_{0}$ for the
constant initial forward variance curve.
\end{assumption}

\subsection{ATM volatility skew }

This subsection derives the ATM volatility skew of the rBergomi and
Bergomi models, as the more realistic ATM volatility skew of the rBergomi
model over the Bergomi model is one of the motivations behind the
introduction of the rBergomi model.

From \citet{bergomi2012stochastic}, we can define the price and the
volatility dynamics of a generic stochastic volatility model as follows:
\begin{equation}
\left\{ \begin{aligned} & dX_{t}=-\frac{1}{2}V_{t}dt+\sqrt{V_{t}}dW_{t}\\
 & d\xi_{t}^{u}=\lambda(t,u,\xi_{t}^{u})dB_{t},
\end{aligned}
\right.\label{eq:general_SV}
\end{equation}
where in particular note that $X_{t}=\xi_t^t=\ln(S_{t})$ is the log-spot, $V_{t}$ is the instantaneous
spot variance, $\xi_{t}^{u}$ is the instantaneous forward variance
for date $u$ observed at time $t$, and $\lambda=(\lambda_{1},\cdots,\lambda_{d})$
is the volatility of forward instantaneous variances which takes values
in $\mathbb{R}^{d}$ where $d$ is the dimension of the Brownian motion
$B$. Note that in this formulation, the covariance between spot and
variance is modelled through the first component of $\lambda$, see \cite{bergomi2012stochastic} for more details.

One can derive the following second-order expression (w.r.t. volatility
of volatility) for the Black-Scholes implied volatility:
\begin{equation}
\sigma_{BS}(k,T)=\hat{\sigma}_{T}^{ATM}+\mathcal{S}_{T}k+\mathcal{C}_{T}k^{2}+\mathcal{O}(\varepsilon^{3})\,,\label{eq:2nd_order_BS_implied_vol}
\end{equation}

where $k=\ln\left(\frac{K}{S_{0}}\right)$, $K$ is the strike and
$\varepsilon$ is a dimensionless scaling factor for the volatility
of variances. The ATM volatility and the two coefficients $\mathcal{S}_{T}$
and $\mathcal{C}_{T}$ are given by
\[
\begin{aligned}\hat{\sigma}_{T}^{ATM} & =\hat{\sigma}_{T}^{VS}\left[1+\frac{\varepsilon}{4v}C^{X\xi}+\frac{\varepsilon^{2}}{32v^{3}}\left(12\left(C^{X\xi}\right){}^{2}-v(v+4)C^{\xi\xi}+4v(v-4)C^{\mu}\right)\right],\\
\mathcal{S}_{T} & =\hat{\sigma}_{T}^{VS}\left[\frac{\varepsilon}{2v^{2}}C^{X\xi}+\frac{\varepsilon^{2}}{8v^{3}}\left(4C^{\mu}v-3(C^{X\xi})^{2}\right)\right],\\
\mathcal{C}_{T} & =\hat{\sigma}_{T}^{VS}\frac{\varepsilon^{2}}{8v^{4}}\left[4C^{\mu}v+C^{\xi\xi}v-6\left(C^{X\xi}\right){}^{2}\right]
\end{aligned}
\]

where $v=\int_{0}^{T}\xi_{0}^{s}ds$ is the total variance to expiration
$T$, $\hat{\sigma}_{T}^{VS}=\sqrt{\frac{v}{T}}=\sqrt{\frac{\int_{0}^{T}\xi_{0}^{s}ds}{T}}$
is the effective volatility, and $C^{X\xi},~C^{\xi\xi},~C^{\mu}$
are autocorrelations \citep{bergomi2012stochastic}:
\begin{itemize}
\item $C_{t}^{X\xi}(\xi)=\int_{t}^{T}ds\int_{s}^{T}du\mu\left(s,u,\xi\right)=\int_{t}^{T}ds\int_{s}^{T}du\frac{\mathbb{E}\left[dX_{s}d\xi_{s}^{u}\right]}{ds}$
is the doubly integrated spot-variance covariance function
\item $C^{X\xi}=C_{0}^{X\xi}(\xi_{0})=\int_{0}^{T}ds\int_{s}^{T}du\frac{\mathbb{E}\left[dX_{s}d\xi_{0}^{u}\right]}{ds}$
\item $C_{t}^{\xi\xi}(\xi)=\int_{t}^{T}ds\int_{s}^{T}du\int_{s}^{T}du'\nu(s,u,u',\xi)=\int_{t}^{T}ds\int_{s}^{T}du\int_{s}^{T}u'\frac{\mathbb{E}\left[d\xi_{s}^{u}d\xi_{s}^{u'}\right]}{ds}$
is the triply integrated variance/variance covariance function
\item $C^{\xi\xi}=C_{0}^{\xi\xi}(\xi_{0})=\int_{0}^{T}dt\int_{s}^{T}du\int_{s}^{T}d{u'}\frac{\mathbb{E}\left[d\xi_{0}^{u}d\xi_{0}^{u'}\right]}{ds}$.
\item $C_{t}^{\mu}(\xi)=\int_{t}^{T}ds\int_{s}^{T}du\mu\left(s,u,\xi\right)\partial_{\xi_{0}^{u}}\left(C_{s}^{X\xi}(\xi)\right)$
is the double time-integral of the instance spot variance covariance
function times the sensitivity of $C_{t}^{X\xi}(\xi)$ with respect
to instantaneous forward variances
\item $C^{\mu}=C_{0}^{\mu}(\xi_{0})=\int_{0}^{T}ds\int_{s}^{T}du\frac{\mathbb{E}\left[dX_{s}d\xi_{0}^{u}\right]}{ds}\partial_{\xi_{0}^{u}}\left(C_{s}^{X\xi}\left(\xi\right)\right)$.
\end{itemize}

\subsubsection{ATM volatility skew in the rBergomi model}
\begin{thm}
In the rBergomi model \eqref{eq:rBergomi}, the ATM volatility skew
$\psi(T)$ satisfies
\begin{equation}
\psi(T)\triangleq\left|\frac{\partial}{\partial_{k}}\sigma_{BS}(k,T)\right|_{k=0}\sim T^{H-\frac{1}{2}}\,.\label{eq:rBergomi_skew}
\end{equation}
\end{thm}
\begin{proof}
We first explicit the autocorrelation functional in the rBergomi model.
Using the fact that $\frac{\mathbb{E}[dX_{t}d\xi_{t}^{u}]}{dt}=\rho\eta\sqrt{2\alpha+1}(u-t)^{\alpha}\sqrt{\xi_{t}^{t}}\xi_{t}^{u}$,
the autocorrelation functionals $C^{X\xi}$ and $C^{\xi\xi}$ are
given by
\[
\begin{aligned}C^{X\xi} & =\int_{0}^{T}ds\int_{s}^{T}du\frac{\mathbb{E}[dX_{s}d\xi_{0}^{u}]}{ds}\\
 & =\rho\eta\sqrt{2\alpha+1}\int_{0}^{T}\sqrt{\xi_{0}^{s}}ds\int_{s}^{T}\xi_{0}^{u}(u-s)^{\alpha}du+\mathcal{O}\left(\varepsilon^{3}\right),\\
C^{\xi\xi} & =\int_{0}^{T}ds\int_{s}^{T}du\int_{s}^{T}du'\frac{\mathbb{E}[d\xi_{0}^{u}d\xi_{0}^{u'}]}{ds}\\
 & =\int_{0}^{T}ds\int_{s}^{T}du\int_{s}^{T}du'\eta^{2}(2\alpha+1)(u-s)^{\alpha}(u'-s)^{\alpha}\xi_{0}^{u}\xi_{0}^{u'}\\
 & =\eta^{2}(2\alpha+1)\int_{0}^{T}ds\left(\int_{0}^{T}\xi_{0}^{u}(u-s)^{\alpha}du\right)^{2}+\mathcal{O}\left(\varepsilon^{4}\right).
\end{aligned}
\]
Then, using the fact that
\[
\begin{aligned}\partial_{\xi_{s}^{u}}(C_{s}^{X\xi}(\xi)) & =\rho\eta\sqrt{2\alpha+1}\left[\int_{s}^{T}dt\sqrt{\xi_{s}^{t}}(u-t)^{\alpha}\textbf{1}_{u>t}+\frac{1}{2\sqrt{\xi_{s}^{u}}}\int_{u}^{T}\xi_{s}^{t}(t-u)^{\alpha}dt\right]\\
 & =\rho\eta\sqrt{2\alpha+1}\left[\int_{s}^{u}dt\sqrt{\xi_{s}^{t}}(u-t)^{\alpha}+\frac{1}{2\sqrt{\xi_{s}^{u}}}\int_{u}^{T}\xi_{s}^{t}(t-u)^{\alpha}dt\right],
\end{aligned}
\]
we obtain
\[
\begin{aligned}C^{\mu}= & \int_{0}^{T}ds\int_{s}^{T}du\frac{\mathbb{E}[dX_{s}d\xi_{0}^{u}]}{dt}\partial_{\xi_{0}^{u}}\left(C_{s}^{X\xi}(\xi)\right)\\
= & \rho^{2}\eta^{2}(2\alpha+1)\int_{0}^{T}\sqrt{\xi_{0}^{s}}ds\int_{s}^{T}(u-s)^{\alpha}du\\
 & \times\left[\int_{s}^{u}\sqrt{\xi_{0}^{t}}\xi_{0}^{u}(u-t)^{\alpha}dt+\frac{\sqrt{\xi_{0}^{u}}}{2}\int_{u}^{T}\xi_{0}^{t}(t-u)^{\alpha}dt\right]+\mathcal{O}\left(\varepsilon^{4}\right).
\end{aligned}
\]
Therefore, using Assumption \textcolor{blue}{\ref{assu:initial_curve}},
we obtain the following explicit first-order approximation:

\[
C^{X\xi}=\rho\eta\sqrt{2H}\int_{0}^{T}\sqrt{\xi_{0}}ds\int_{s}^{T}\xi_{0}\left(u-s\right){}^{\alpha}du\mathcal{+O}\left(\varepsilon^{3}\right)\approx C_{H}\xi_{0}^{\frac{3}{2}}T^{H+\frac{3}{2}}\,,
\]
where $C_{H}$ is a constant depending on $H$. We are then able to
compute the first-order approximations of the three correlation values
$C^{X\xi},~C^{\xi\xi},~C^{\mu}$ explicitly. The first-order approximation
of $\sigma_{BS}(k,T)$ can be written as follows:
\[
\begin{aligned}\sigma_{BS}(k,T) & =\hat{\sigma}_{T}^{VS}+\frac{1}{4v}C^{x\xi}\hat{\sigma}_{T}^{VS}\varepsilon+\frac{1}{2v^{2}}C^{X\xi}\hat{\sigma}_{T}^{VS}\varepsilon k\\
 & =\hat{\sigma}^{VS}+\left(\frac{1}{4v}+\frac{k}{2v^{2}}\right)C_{H}\xi_{0}^{\frac{3}{2}}T^{H+\frac{3}{2}}\hat{\sigma}_{T}^{VS}\varepsilon,
\end{aligned}
\]
Thus, the ATM volatility skew generated by the rBergomi model satisfies
Equation \ref{eq:rBergomi_skew}, which is consistent with empirical
evidence (see for example, \citet{gatheral2018volatility}).
\end{proof}

\subsubsection{ATM volatility skew in the two-factor Bergomi model}

We now compare this result to the volatility skew in the classical
two-factor Bergomi model.
\begin{thm}
In the two-factor Bergomi model , the ATM volatility skew satisfies

\begin{equation}
\psi(T)\sim\frac{C_{1}\left(\kappa_{X}T-1+e^{-\kappa_{X}T}\right)}{T^{2}}+\frac{C_{2}\left(\kappa_{Y}T-1+e^{-\kappa_{Y}T}\right)}{T^{2}}\label{eq:2Dvolskew}
\end{equation}
\end{thm}
\begin{proof}
The Brownian motions $W^{S},W^{X},W^{Y}$ can be decomposed as:
\[
\begin{aligned} & W^{S}=W^{1}\\
 & W^{X}=\rho_{SX}W^{1}+\sqrt{1-\rho_{SX}^{2}}W^{2}\\
 & W^{Y}=\rho_{SY}W^{1}+\chi\sqrt{1-\rho_{SY}^{2}}W^{2}+\sqrt{(1-\chi^{2})(1-\rho_{SY}^{2})}W^{3},
\end{aligned}
\]
where $W^{1},W^{2},W^{3}$ are three independent Brownian motions
and $\chi\triangleq\frac{\rho_{XY}-\rho_{SX}\rho_{SY}}{\sqrt{1-\rho_{SX}^{2}}\sqrt{1-\rho_{SY}^{2}}}$.
Thus the volatilities of variance $\lambda=(\lambda_{1},\lambda_{2},\lambda_{3})$
in the general formulation \eqref{eq:general_SV} can be written as:
\[
\begin{aligned} & \lambda_{1}(t,u,\xi)=\alpha_{\theta}\omega\xi_{0}^{u}\left[\left(1-\theta\right)\rho_{SX}e^{-\kappa_{X}(u-t)}+\theta\rho_{SY}e^{-\kappa_{Y}(u-t)}\right],\\
 & \lambda_{2}(t,u,\xi)=\alpha_{\theta}\omega\xi_{0}^{u}\left[\left(1-\theta\right)\sqrt{1-\rho_{SX}^{2}}e^{-\kappa_{X}(u-t)}+\theta\chi\sqrt{1-\rho_{SY}^{2}}e^{-\kappa_{Y}(u-t)}\right],\\
 & \lambda_{3}(t,u,\xi)=\alpha_{\theta}\omega\xi_{0}^{u}\theta\sqrt{\left(1-\chi^{2}\right)\left(1-\rho_{SY}^{2}\right)}e^{-\kappa_{Y}(u-t)},
\end{aligned}
\]
or equivalently:
\[
\lambda_{i}(t,u,\xi)=\alpha_{\theta}\omega\xi_{0}^{u}\left(\omega_{iX}e^{-\kappa_{X}(u-t)}+\omega_{iY}e^{-\kappa_{Y}(u-t)}\right),
\]
where
\[
\begin{aligned} & \left(\omega_{iX}\right)_{i=1,2,3}\triangleq\left(\left(1-\theta\right)\rho_{SX},\left(1-\theta\right)\sqrt{1-\rho_{SX}^{2}},0\right){}^{\top},\\
 & \left(\omega_{iY}\right)_{i=1,2,3}\triangleq\left(\theta\rho_{SY},\theta\chi\sqrt{1-\rho_{SY}^{2}},\theta\sqrt{(1-\chi^{2})(1-\rho_{SY}^{2})}\right){}^{\top}.
\end{aligned}
\]

The corresponding covariances can be expressed similarly as:
\[
\begin{aligned}C^{X\xi}= & \int_{0}^{T}du\int_{0}^{u}dt\sqrt{\xi_{0}^{t}}\lambda_{1}\left(t,u,\xi_{0}\right)\\
= & \alpha_{\theta}\omega\left[(1-\theta)\rho_{SX}\int_{0}^{T}du\xi_{0}^{u}\int_{0}^{u}dt\sqrt{\xi_{0}^{t}}e^{-\kappa_{X}(u-t)}+\theta\rho_{SY}\int_{0}^{T}du\xi_{0}^{u}\int_{0}^{u}dt\sqrt{\xi_{0}^{t}}e^{-\kappa_{Y}(u-t)}\right]\\
C^{\xi\xi}= & \sum_{i=1}^{3}\int_{0}^{T}ds\left(\int_{s}^{T}du\lambda_{i}\left(s,u,\xi_{0}\right)\right){}^{2}\\
= & \alpha_{\theta}^{2}\omega^{2}\sum_{i=1}^{3}\int_{0}^{T}ds\left(\omega_{iX}\int_{s}^{T}du\xi_{0}^{u}e^{-\kappa_{X}(u-s)}+\omega_{iY}\int_{s}^{T}du\xi_{0}^{u}e^{-\kappa_{Y}(u-s)}\right){}^{2}\\
C^{\mu}= & \int_{0}^{T}ds\int_{s}^{T}du\sqrt{\xi_{0}^{s}}\lambda_{1}\left(s,u,\xi_{0}\right)\left(\frac{1}{2\sqrt{\xi_{0}^{u}}}\int_{u}^{T}dt\lambda_{1}\left(u,t,\xi_{0}\right)+\int_{s}^{u}dr\sqrt{\xi_{0}^{r}}\partial_{\xi_{0}^{u}}\lambda_{1}\left(r,u,\xi\right)\right).
\end{aligned}
\]

Using once again Assumption \ref{assu:initial_curve}, we obtain
\[
\begin{aligned}C^{X\xi}= & \alpha_{\theta}\omega\xi_{0}^{\frac{3}{2}}T^{2}\left(\omega_{1X}\mathcal{J}(\kappa_{X}T)+\omega_{1Y}\mathcal{J}(\kappa_{Y}T)\right)\\
C^{\xi\xi}= & \alpha_{\theta}^{2}\omega\xi_{0}^{2}T^{3}\left(\omega_{0}+\omega_{X}\mathcal{I}\left(\kappa_{X}T\right)+\omega_{Y}\mathcal{I}\left(\kappa_{Y}T\right)+\omega_{XX}\mathcal{I}\left(2\kappa_{X}T\right)+\omega_{YY}\mathcal{I}\left(2\kappa_{Y}T\right)+\omega_{XY}\mathcal{I}\left(\left(\kappa_{X}+\kappa_{Y}\right)T\right)\right),
\end{aligned}
\]

where
\begin{align*}
\omega_{0} & =\sum_{i=1}^{3}\left(\frac{\omega_{iX}}{\kappa_{X}T}+\frac{\omega_{iY}}{\kappa_{Y}T}\right)^{2},\omega_{X}=-2\sum_{i=1}^{3}\frac{\omega_{iX}}{\kappa_{X}T}\left(\frac{\omega_{iX}}{\kappa_{X}T}+\frac{\omega_{iY}}{\kappa_{Y}T}\right),\omega_{Y}=-2\sum_{i=1}^{3}\frac{\omega_{iY}}{\kappa_{Y}T}\left(\frac{\omega_{iX}}{\kappa_{X}T}+\frac{\omega_{iY}}{\kappa_{Y}T}\right),
\end{align*}

\[
\omega_{XX}=\sum_{i=1}^{3}\frac{\omega_{iX}^{2}}{\kappa_{X}^{2}T^{2}},\ \ \ \ \omega_{YY}=\sum_{i=1}^{3}\frac{\omega_{iY}^{2}}{\kappa_{Y}^{2}T^{2}},\ \ \ \ \omega_{XY}=2\sum_{i=1}^{3}\frac{\omega_{iX}\omega_{iY}}{\kappa_{X}\kappa_{Y}T^{2}},
\]

and
\[
\mathcal{I}(z)=\frac{1-e^{-z}}{z},~\mathcal{J}(z)=\frac{z-1+e^{-z}}{z^{2}},~\mathcal{K}(z)=\frac{1-e^{-z}-ze^{-z}}{z^{2}},~\mathcal{H}(z)=\frac{\mathcal{J}(z)-\mathcal{K}(z)}{z}.
\]
Similarly, we have $C^{\mu}=\alpha_{\theta}^{2}\omega^{2}\xi_{0}^{2}T^{3}\left(C_{1}^{\mu}+C_{2}^{\mu}\right)$,
where the coefficients
\[
\begin{aligned} & C_{1}^{\mu}=\frac{1}{2}\omega_{1X}^{2}\mathcal{H}\left(\kappa_{X}T\right)+\frac{1}{2}\omega_{1Y}^{2}\mathcal{H}\left(\kappa_{Y}T\right)-\omega_{1X}\omega_{1Y}\frac{\mathcal{J}\left(\kappa_{Y}T\right)-\mathcal{J}\left(\kappa_{X}T\right)}{\left(\kappa_{X}+\kappa_{Y}\right)T},\\
 & C_{2}^{\mu}=\omega''_{X}\mathcal{J}\left(\kappa_{X}T\right)+\omega''_{Y}\mathcal{J}\left(\kappa_{Y}T\right)+\omega''_{XX}\mathcal{J}\left(2\kappa_{X}T\right)+\omega''_{YY}\mathcal{J}\left(2\kappa_{Y}T\right)+\omega''_{XY}\mathcal{J}\left(\left(\kappa_{X}+\kappa_{Y}\right)T\right),
\end{aligned}
\]
and
\begin{eqnarray*}
\omega''_{X}=\frac{\omega_{1X}^{2}}{\kappa_{X}T}+\frac{\omega_{1X}\omega_{1Y}}{\kappa_{Y}T}, & \omega''_{Y}=\frac{\omega_{1Y}^{2}}{\kappa_{Y}T}+\frac{\omega_{1X}\omega_{1Y}}{\kappa_{Y}T},\\
\omega_{XX}''=-\frac{\omega_{1X}^{2}}{\kappa_{X}T}, & \omega_{YY}''=-\frac{\omega_{1Y}^{2}}{\kappa_{Y}T}, & \omega_{XY}''=-\frac{\omega_{1X}\omega_{1Y}}{\kappa_{X}T}-\frac{\omega_{1X}\omega_{1Y}}{\kappa_{Y}T}.
\end{eqnarray*}
Since $C^{X\xi}\sim T^{2}\left(C_{1}\cdot\frac{\kappa_{X}T-1+e^{-\kappa_{X}T}}{\left(\kappa_{X}T\right){}^{2}}+C_{2}\cdot\frac{\kappa_{Y}T-1+e^{-\kappa_{Y}T}}{\left(\kappa_{Y}T\right){}^{2}}\right)$
and $C_{1},C_{2}$ are constants, we can derive the term structure
of the ATM volatility skew as in equation (\ref{eq:2Dvolskew}) with
first order in $\varepsilon$.
\end{proof}
However, this result derived for the Bergomi model by the Bergomi-Guyon
expansion \citep{bergomi2012stochastic} is inconsistent with empirical
evidence, see for example \citet{bayer2016pricing}. This suggests
that the power-law kernel of the forward variance curve in the rBergomi
model will lead to more realistic and accurate pricing and hedging
results than the exponential kernel of the forward variance curve
in the Bergomi model.

\section{Markovian representation of the rough Bergomi model\label{sec:Markovian}}

The purpose of this section is to establish the infinite-dimensional
affine nature and Markovianity of the rBergomi model.

\begin{myDef}\label{def:OU} An Ornstein-Uhlenbeck (O-U) process
$Y^x_{t}$ is the solution of the following stochastic differential
equation (SDE):
\begin{equation}
dY_{t}^x=x(a-Y_{t}^x)dt+\sigma dB_{t}\label{eq:OU_SDE}
\end{equation}
where $x>0$ is the mean-reversion speed, $a>0$ is the mean-reversion
level, and $B_{s}$ is a standard Brownian motion. Its strong solution
is explicitly given by
\begin{equation}
Y_{t}^x=Y_{0}+\sigma\int_{0}^{t}e^{-x(t-s)}dB_{s}\label{eq:OU_solution}.
\end{equation}
\end{myDef}

\begin{assumption}
In the rest of the paper, we always assume that
\begin{align}
a & \triangleq Y_{0}\label{eq:OU0}\\
\sigma & \triangleq\eta\sqrt{2\alpha+1}\label{eq:OUsigma}
\end{align}
 where $\eta$ and $\alpha$ come from the Definition \ref{def:rBergomi}
of the rBergomi model (see \citealt{bayer2016pricing}).
\end{assumption}

\begin{myDef} \label{def:measure} Without loss of generality,
we define the sigma-finite measure $\mu(dx)$ on $(0,\infty)$ as
$\mu(dx)=\frac{dx}{x^{\frac{1}{2}+H}\Gamma(\frac{1}{2}-H)}$.
\end{myDef}

\subsection{Volterra-type integral as a functional of a Markov process}
\begin{thm}
\label{thm:integral_representation}Using Definitions \ref{def:OU}
and \ref{def:measure}, the Volterra type integral $\tilde{X_{t}}\triangleq\int_{0}^{t}(t-s)^{H-\frac{1}{2}}dB_{s}$
in the rBergomi model has the Markovian representation
\begin{equation}
\sigma\tilde{X_{t}}=\int_{0}^{\infty}\left(Y^x_{t}-Y^x_{0}\right)\mu(dx)\,.\label{eq:markovian_rep}
\end{equation}
\end{thm}
\begin{proof}
the Laplace transform of the measure $\mu$ in Definition \ref{def:measure}
is
\[
\mathcal{L}(\mu)(\tau)=\int_{0}^{\infty}e^{-\tau x}\mu(dx)=\int_{0}^{\infty}\frac{e^{-\tau x}x^{-\frac{1}{2}-H}}{\Gamma\left(\frac{1}{2}-H\right)}dx=\tau^{H-\frac{1}{2}}\,,
\]
which can be recognised as the power-law kernel in the Volterra type
integral. Consequently, we have $\sigma\tilde{X}_{t}=\int_{0}^{t}\int_{0}^{\infty}\sigma e^{-x(t-s)}\mu(dx)dB_{s}$,
and using Fubini's stochastic theorem \citep{protter2005stochastic},
we obtain $\sigma\tilde{X_{t}}=\int_{0}^{\infty}\int_{0}^{t}\sigma e^{-x(t-s)}dB_{s}\mu(dx)$.
From Definition \ref{def:OU}, where $\int_{0}^{t}\sigma e^{-x(t-s)}dB_{s}=Y^x_{t}-Y_{0}$,
we obtain the Markovian representation given by equation \ref{eq:markovian_rep}.
\end{proof}
\begin{thm}
\label{thm:OUaffine} The O-U process \eqref{eq:OU_solution} has
the affine structure
\begin{eqnarray*}
\mathbb{E}\left[e^{\int_{0}^{\infty}Y^x_{t}\mu(dx)}\mid\mathcal{F}_{s}\right] & = & e^{\frac{\sigma^{2}}{2}\int_{0}^{t-s}\left(\int_{0}^{\infty}e^{-sx}\mu(dx)\right)^{2}ds+\int_{0}^{\infty}Y^x_{s}e^{-(t-s)x}\mu(dx)}.
\end{eqnarray*}
\end{thm}
\begin{proof}
From Fubini's stochastic theorem, $\int_{0}^{\infty}Y^x_{t}\mu(dx)$
is Gaussian under the filtration $\mathcal{F}_{s}$ for $0\leq s\leq t$,
with mean
\[
\mathbb{E}\left[\int_{0}^{\infty}Y^x_{t}\mu(dx)\mid\mathcal{F}_{s}\right]=\int_{0}^{\infty}Y^x_{s}e^{-(t-s)x}\mu\left(dx\right).
\]
Furthermore, using It\=o's isometry, we have the conditional variance:
\[
\begin{aligned}\mathrm{Var}\left(\int_{0}^{\infty}Y^x_{t}\mu(dx)\mid\mathcal{F}_{s}\right) & =\sigma^{2}\int_{s}^{t}\left(\int_{0}^{\infty}e^{-(t-s)x}\mu(dx)\right){}^{2}ds\\
 & =\sigma^{2}\int_{0}^{t-s}\left(\int_{0}^{\infty}e^{-sx}\mu\left(dx\right)\right){}^{2}ds.
\end{aligned}
\]
Thus
\[
\begin{aligned}\mathbb{E}\left[e^{\int_{0}^{\infty}Y^x_{t}\mu(dx)}\mid\mathcal{F}_{s}\right] & =e^{\frac{1}{2}\mathrm{Var}\left(\int_{0}^{\infty}Y^x_{t}\mu(dx)\mid\mathcal{F}_{s}\right)+\mathbb{E}\left[\int_{0}^{\infty}Y^x_{t}\mu(dx)\mid\mathcal{F}_{s}\right]}\\
 & =e^{\frac{\sigma^{2}}{2}\int_{0}^{t-s}\left(\int_{0}^{\infty}e^{-sx}\mu(dx)\right){}^{2}ds+\int_{0}^{\infty}Y^x_{s}e^{-(t-s)x}\mu(dx)}.
\end{aligned}
\]
\end{proof}

\subsection{Affine structure in the rBergomi model}

From Definition \ref{def:rBergomi} and Theorem \ref{thm:integral_representation},
the rBergomi model can be rewritten in the following form:
\[
\left\{ \begin{aligned}dX_{t} & =-\frac{1}{2}V_{t}dt+\sqrt{V_{t}}dW_{t}\\
\log\frac{V_{t}}{\xi_{0}} & =\int_{0}^{\infty}(Y^x_{t}-Y_{0})\mu(dx)\,,
\end{aligned}
\right.
\]
where $X_{t}$ is the log stock price, $\xi_{0}$ is the initial flat
forward variance curve, and $W,B$ are two Brownian motions with correlation
$d\langle W,B\rangle_{t}=\rho dt$ and $\rho\in[-1,1]$.

Our aim is now to write the log stock price $X_{t}$ in affine form
as the first coordinate of an infinite-dimensional affine process.
To do so, we introduce the following symmetric nonnegative tensor:
\[
L^{1}(\mu)\otimes_{s}L^{1}(\mu)=\left\{ y^{\otimes2}:y\in L^{1}(\mu)\right\} \subset L^{1}(\mu)^{\otimes2}\subset L^{1}(\mu^{\otimes2})\,.
\]

Let $\Pi_{t}=(i\otimes1)\left(Y^x_{t}\right)^{\otimes2}\in iL^{1}(\mu)\otimes_{s}L^{1}(\mu)$.
The relation $\left(\int_{0}^{\infty}Y^x_{t}\mu(dx)\right)^{2}=\int_{0}^{\infty}(i\otimes1)(Y^x_{t})^{\otimes2}\mu^{\otimes2}(dx)$
holds.

Therefore, the log stock price dynamics can be written as
\[
\begin{aligned}dX_{t} & =\sqrt{\xi_{0}}\cdot\left(\mathcal{E}^{\frac{\int_{0}^{\infty}\Pi_{t}\mu^{\otimes2}(dx)}{4}}dW_{t}-\frac{1}{2}\mathcal{E}^{\int_{0}^{\infty}Y^x_{t}\mu(dx)}\right)\\
 & =\sqrt{\xi_{0}}e^{\frac{\int_{0}^{\infty}\Pi_{t}\mu^{\otimes2}(dx)}{4}}e^{-\frac{\eta^{2}}{4}t^{2\alpha+1}}dW_{t}-\frac{\sqrt{\xi_{0}}}{2}e^{\int_{0}^{\infty}Y^x_{t}\mu(dx)}e^{-\frac{\eta^{2}}{2}t^{2\alpha+1}}dt,
\end{aligned}
\]
 where $\mathcal{E}$ is the Dol\'eans-Dade stochastic exponential.
\begin{thm}
\label{thm:Pi_affine} The process $\Pi_{t}=(i\otimes1)(Y^x_{t})^{\otimes2}$
satisfies the affine structure

\begin{equation}
\mathbb{E}\left[e^{\int_{0}^{\infty}\Pi_{t}\mu^{\otimes2}(dx)}\mid\mathcal{F}_{s}\right]=e^{\Phi_{1}+\Phi_{2}}\label{eq:Phi}
\end{equation}

where
\begin{align}
\Phi_{1} & \triangleq-\frac{1}{2}\log\left(1-2\int_{0}^{t-s}\left(\int_{0}^{\infty}e^{-sx}\mu\left(dx\right)\right){}^{2}\right)ds\label{eq:Phi1}\\
\Phi_{2} & \triangleq\frac{\sigma^{2}\left(e^{-(t-s)x}\right){}^{\otimes2}}{1-2\int_{0}^{t-s}\left(\int_{0}^{\infty}e^{-sx}\mu\left(dx\right)\right){}^{2}}ds\,.\label{eq:Phi2}
\end{align}
\end{thm}
\begin{proof}
From Fubini's stochastic theorem, $\frac{\int_{0}^{\infty}Y^x_{t}\mu(dx)}{\sigma\sqrt{\int_{0}^{t-s}(\int_{0}^{\infty}e^{-sx}\mu(dx))^{2}ds}}$
is Gaussian under the filtration $\mathcal{F}_{s}$ for $0\leq s\leq t$,
with conditional mean
\[
\mathbb{E}\left[\frac{\int_{0}^{\infty}Y^x_{t}\mu\left(dx\right)}{\sigma\sqrt{\int_{0}^{t-s}\left(\int_{0}^{\infty}e^{-sx}\mu(dx)\right){}^{2}ds}}\mid\mathcal{F}_{s}\right]=\frac{\int_{0}^{\infty}Y^x_{s}e^{-(t-s)x}\mu\left(dx\right)}{\sigma\sqrt{\int_{0}^{t-s}\left(\int_{0}^{\infty}e^{-sx}\mu\left(dx\right)\right){}^{2}ds}}
\]
and conditional variance
\[
\mathrm{Var}\left(\frac{\int_{0}^{\infty}Y^x_{t}\mu\left(dx\right)}{\sigma\sqrt{\int_{0}^{t-s}(\int_{0}^{\infty}e^{-sx}\mu(dx))^{2}ds}}\mid\mathcal{F}_{s}\right)=1.
\]
Then, the random variable defined as
\[
\frac{\int_{0}^{\infty}\Pi_{t}\mu^{\otimes2}\left(dx\right)}{\sigma^{2}\int_{0}^{t-s}\left(\int_{0}^{\infty}e^{-sx}\mu\left(dx\right)\right){}^{2}ds}=\left(\frac{\int_{0}^{\infty}Y^x_{t}\mu\left(dx\right)}{\sigma\sqrt{\int_{0}^{t-s}(\int_{0}^{\infty}e^{-sx}\mu\left(dx\right))^{2}ds}}\right){}^{2}
\]
is a noncentral $\chi^{2}$ distribution with one degree of freedom
and noncentrality parameter
\[
\frac{\left(\int_{0}^{\infty}Y^x_{s}e^{-(t-s)x}\mu\left(dx\right)\right){}^{2}}{\sigma^{2}\int_{0}^{t-s}\left(\int_{0}^{\infty}e^{-sx}\mu\left(dx\right)\right){}^{2}ds}=\frac{\int_{0}^{\infty}\Pi_{s}\left(e^{-(t-s)x}\right){}^{\otimes2}\mu^{\otimes2}\left(dx\right)}{\sigma^{2}\int_{0}^{t-s}\left(\int_{0}^{\infty}e^{-sx}\mu\left(dx\right)\right){}^{2}ds}.
\]

Thus the formulas \eqref{eq:Phi1} and \eqref{eq:Phi2} for $\Phi_{1}$
and $\Phi_{2}$ follow from the characteristic function of the noncentral
$\chi^{2}$ distribution, which concludes the proof.
\end{proof}
\begin{cor}
\label{cor:rBergomi_Markovian}The rBergomi model is an infinite-dimensional
Markovian process.
\end{cor}
\begin{proof}
This corollary follows from Theorem \ref{thm:Pi_affine} which exhibits
that the rBergomi model has an exponential-affine dependence on $x$,
hence the model is Markovian in each dimension.
\end{proof}

\section{Approximation of the rough Bergomi model by the aBergomi model\label{sec:rBergomi_approximation}}

In this Section, we first introduce the aBergomi model which is used
to approximate the rBergomi model \eqref{eq:rBergomi}. After that,
we will demonstrate the existence and uniqueness of the solution of
this aBergomi model. We also prove that the aBergomi model is well-defined
and the solution of the aBergomi model converges to that of the rBergomi
model when the number of terms $n$ in the aBergomi model goes to
infinity. At the same time, we show that the rBergomi model inherits
the affine structure of the Bergomi model.

Since the rBergomi model can be represented by
\[
\left\{ \begin{aligned}dS_{t} & =S_{t}\sqrt{V_{t}}dW_{t}\\
\log\left\{ \frac{V_{t}}{\xi_{0}}\right\}  & =\int_{0}^{\infty}\sigma\int_{0}^{t}e^{-x(t-s)}dB_{s}\mu(dx)
\end{aligned}
\right.
\]
and the $n$-term Bergomi model with the same Brownian motion in the
variance process can be represented by
\begin{equation}
\left\{ \begin{aligned}dS_{t} & =S_{t}\sqrt{V_{t}}dW_{t}\\
\log\left\{ \frac{V_{t}}{\xi_{0}}\right\}  & =\int_{0}^{t}\left(\sum_{i=1}^{n}\alpha_{i}e^{-\kappa_{i}\left(t-s\right)}\right)dB_{s},
\end{aligned}
\right.\label{eq:aBergomi}
\end{equation}
we can view the rBergomi model as a continuous infinite-term Bergomi
model under the measure $\mu(\cdot)$, in which the mean-reversion
speed $x$ has been integrated from $0$ to $\infty$, with the Brownian
motion $B_{s}$. We can therefore approximate the rBergomi model by
a $n$-term exponential kernel $K_{\text{exp}}=\sum_{i=1}^{n}\alpha_{i}e^{-\kappa_{i}(t-s)}$
instead of the power kernel $K_{\text{pow}}=\sqrt{2\alpha+1}(t-s)^{\alpha}$
of the Volterra process in the rBergomi model.

Following equation (\ref{eq:aBergomi}), after approximating
the exponential kernel $K(\tau)=\int_{0}^{\infty}e^{-x\tau}\mu(dx)$
by the kernel $K^{n}(\tau)=\sum_{i=1}^{n}\alpha_{i}^{n}e^{-\tau x_{i}^{n}}$
, we can rewrite the aBergomi model \eqref{eq:aBergomi} as follows:

\begin{equation}
\left\{ \begin{aligned}dS_{t}^{n} & =S_{t}^{n}\sqrt{V_{t}^{n}}dW\\
\log\left\{ \frac{V_{t}^{n}}{\xi_{0}}\right\}  & =\sum_{i=1}^{n}\alpha_{i}^{n}V_{t}^{n,i}\\
dV_{t}^{n,i} & =-x_{i}^{n}\left(a-V_{t}^{n}\right)dt+\sigma dB_{t}~~a=Y_{0},~\sigma=\eta\sqrt{2\alpha+1},
\end{aligned}
\right.\label{eq:aBergomi_S1}
\end{equation}
where $(\alpha_{i}^{n})_{1\leq i\leq n}$ are positive weights, $(x_{i}^{n})_{1\leq i\leq n}$
are mean-reverting speeds, and $\langle W,B\rangle_{t}=\rho dt$,
with initial conditions $S_{0}^{n}=S_{0}=1$ and $V_{0}^{n,i}=V_{0}=0$.

\subsection{Existence and uniqueness of $(S^{n},V^{n})$}

We rewrite $V^{n}$ in \eqref{eq:aBergomi_S1} as the following stochastic
equation
\begin{equation}
\log\left(\frac{V_{t}^{n}}{\xi_{0}}\right)=\sigma\int_{0}^{t}K^{n}\left(t-s\right)dB_{s}.\label{eq:Vn}
\end{equation}

\begin{thm}
Under the conditions of the model \eqref{eq:aBergomi_S1}, there exists
a unique strong non-negative solution $V^{n}$ to equation \eqref{eq:Vn}.
\end{thm}
\begin{proof}
\citet{oksendal1993stochastic} implies that there exists a unique
strong non-negative solution $V^{n}$ to equation \eqref{eq:Vn} under
the conditions of the model \eqref{eq:aBergomi_S1} .
\end{proof}
Then the strong existence and uniqueness of $(S^{n},V^{n})$ follows,
along with its Markovianity w.r.t. the spot price $S^{n}$ and the
factors $V^{n,i}$ for $i\in\{1,\cdots,n\}$.

\subsection{Convergence of $(S^{n},V^{n})$ to $(S,V)$}

To prove that the solution of the aBergomi model $\left(S^{n},V^{n}\right)$
converges to the solution of the rBergomi model $(S,V)$, we need
to choose a suitable $K^{n}(\tau)=\sum_{i=1}^{n}\alpha_{i}^{n}e^{-x_{i}^{n}\tau}$
to approximate $K(\tau)=\tau^{H-\frac{1}{2}}$. When $n\rightarrow+\infty$,
$(V^{n})_{n\geq1}\rightarrow V$ (see \citealt{carmona2000approximation},
\citealt{muravlev2011representation}, \citealt{harms2019affine}).
\begin{thm}
\label{thm:kernel_convergence} There exist weights $(\alpha_{i}^{n})_{1\leq i\leq n}>0$
and mean reversion speeds $(x_{i}^{n})_{1\leq i\leq n}>0$, such that
$\parallel K^{n}-K\parallel_{2,T}\rightarrow0$, where $\|\cdot\|_{2,T}$
is the $L^{2}([0,T],\mathbb{R})$ norm.
\end{thm}
The proof of this theorem is in the Appendix.

Applying the previous computations and the Kolmogorov tightness criterion,
we can get that the sequence $\left(S^{n},V^{n}\right)$ is tight
for the uniform topology and the limit satisfies the model \eqref{eq:aBergomi_S1}.

\subsection{Affine structure of the aBergomi model}

In this section, we detail the affine property of the aBergomi model.
\begin{thm}
The process $V^{n}$ (equation \eqref{eq:Vn}) has the following affine
structure
\[
\mathbb{E}\left[V_{t}^{n}\mid\mathcal{F}_{s}\right]=\xi_{0}\exp\left\{ \frac{\sigma^{2}}{2}\sum_{i=1}^{n}\alpha_{i}^{n}\left(\frac{1}{x_{i}^{n}}-\frac{e^{-(t-s)x_{i}^{n}}}{x_{i}^{n}}\right)+\sum_{i=1}^{n}V_{s}^{n,i}\alpha_{i}^{n}e^{-(t-s)x_{i}^{n}}\right\}
\]
\end{thm}
\begin{proof}
Using Theorem \ref{thm:OUaffine}, we have
\[
\begin{aligned}\mathbb{E}\left[V_{t}^{n}\mid\mathcal{F}_{s}\right] & =\xi_{0}\exp\left\{ \frac{\sigma^{2}}{2}\int_{0}^{t-s}\left(K^{n}(s)\right){}^{2}ds+\sum_{i=1}^{n}V_{s}^{n,i}\alpha_{i}^{n}e^{-(t-s)x_{i}^{n}}\right\} \\
 & =\xi_{0}\exp\left\{ \frac{\sigma^{2}}{2}\int_{0}^{t-s}\left(\sum_{i=1}^{n}\alpha_{i}^{n}e^{-sx_{i}^{n}}\right)ds+\sum_{i=1}^{n}V_{s}^{n,i}\alpha_{i}^{n}e^{-(t-s)x_{i}^{n}}\right\} \\
 & =\xi_{0}\exp\left\{ \frac{\sigma^{2}}{2}\sum_{i=1}^{n}\alpha_{i}^{n}\left(\frac{1}{x_{i}^{n}}-\frac{e^{-(t-s)x_{i}^{n}}}{x_{i}^{n}}\right)+\sum_{i=1}^{n}V_{s}^{n,i}\alpha_{i}^{n}e^{-(t-s)x_{i}^{n}}\right\} ,
\end{aligned}
\]
Similarly we can derive the affine structure of $S^{n}$ by Theorem
\ref{thm:Pi_affine}.
\end{proof}

\section{Numerical method\label{sec:numerical_method}}

In this section, we first introduce the hybrid scheme and algorithm
to approximate an rBergomi model by an aBergomi model. And then, we
compare the simulated volatilities of both models. To demonstrate
the approximation accuracy and efficiency, we investigate the RMSE
of simulated results for different number of terms and number of time
steps in numerical tests. By some improved algorithms, we observe
that 25-term O-U process and 100 time steps can produce a good output,
with reliable outcomes and fast calculation speed under 20000 Monte Carlo paths.

\subsection{Hybrid scheme for simulation}

Recalling equation \eqref{eq:rBergomi}, the rough Bergomi model with
time horizon $T>0$ under an equivalent martingale measure identified
with $\mathbb{P}$ can be written as:

\begin{equation}
\left\{ \begin{aligned}dS_{t} & =S_{t}\sqrt{V_{t}}dW_{t}\\
\frac{d\xi_{s}^{t}}{\xi_{s}^{t}} & =\eta\sqrt{2\alpha+1}(t-s)^{\alpha}dB_{s},
\end{aligned}
\right.\label{eq:rBergomi2}
\end{equation}
where $W,B$ are two standard Brownian motions with correlation $\rho$.
We recall Assumption \ref{assu:initial_curve} that the forward variance
curve $\xi_{0}^{t}$ is flat for all $t\in[0,T]$ : $\xi_{0}^{t}=\xi_{0}>0$
Thus, the spot variance $V_{t}$ in Equation \eqref{eq:rBergomi2}
is given by
\[
V_{t}=\xi_{0}\exp\left(\eta\sqrt{2\alpha+1}\int_{0}^{t}(t-s)^{\alpha}dB_{s}-\frac{\eta^{2}}{2}t^{2\alpha+1}\right).
\]

To simulate the Volterra-type integral $\tilde{X}=\sqrt{2\alpha+1}\int_{0}^{t}(t-s)^{\alpha}dB_{s}$,
we apply the hybrid scheme proposed in \citet{bennedsen2017hybrid},
which approximates the kernel function of the Brownian semi-stationary
processes by a Wiener integrals of the power function at $t=s$ and
a Riemann sum elsewhere.

Let $\left(\Omega,\mathcal{F},\left(\mathcal{F}_{t}\right)_{t\in\mathbb{R}},\mathbb{P}\right)$
be a filtered probability space which supports a standard Brownian
motion $W=\left(W_{t}\right)_{t\in\mathbb{R}}$. We consider a Brownian
semi-stationary process (\textit{Bss}):
\begin{equation}
\bar{X}_{t}=\int_{-\infty}^{t}g(t-s)\sigma_{s}dW_{s}~~t\in\mathbb{R}\,,\label{eq:Bss}
\end{equation}
where $\sigma=(\sigma_{t})_{t\in\mathbb{R}}$ is an $(\mathcal{F}_{t})_{t\in\mathbb{R}}$-predictable
process which captures the stochastic volatility of $\bar{X}$ and
$g:(0,\infty)\rightarrow[0,\infty)$ is a Borel-measurable kernel
function. We assume that $\mathbb{E}\left[\sigma_{t}^{2}\right]<\infty$
for all $t\in\mathbb{R}$ and the process is covariance-stationary,
namely
\[
\begin{aligned}\mathbb{E}\left[\sigma_{s}\right] & =\mathbb{E}\left[\sigma_{t}\right]\\
\mathrm{cov}\left(\sigma_{s},\sigma_{t}\right) & =\mathrm{cov}\left(\sigma_{0},\sigma_{|s-t|}\right),~~s,t\in\mathbb{R}.
\end{aligned}
\]
These assumptions imply that $\bar{X}$ is covariance-stationary.
However, the process $\bar{X}$ need not be strictly stationary.
\begin{assumption}
\label{assu:kernel_assumptions}The assumptions concerning the kernel
function $g$ are as follows:
\begin{description}
\item [{(A1)}] For some $\alpha\in\left(-\frac{1}{2},\frac{1}{2}\right)\backslash\{0\}$,
\[
g(x)=x^{\alpha}L_{g}(x),~~x\in(0,1],
\]
where $L_{g}:(0,1]\rightarrow[0,\infty)$ is continuously differentiable,
slowly varying at 0 and bounded away from 0. Moreover, there exists
a constant $C>0$ such that the derivative $L'_{g}$ of $L_{g}$ satisfies
\[
|L'_{g}(x)|\leq C\left(1+\frac{1}{x}\right),~~x\in(0,1].
\]
\item [{(A2)}] The function $g$ is continuously differentiable on $(0,\infty)$,
with derivative $g'$ that is ultimately monotonic and also satisfies
$\int_{1}^{\infty}g'(x)^{2}dx<\infty$.
\item [{(A3)}] For some $\beta\in\left(-\infty,-\frac{1}{2}\right)$,
\[
g(x)=\mathcal{O}\left(x^{\beta}\right),\;\;x\rightarrow\infty\,.
\]
\end{description}
\end{assumption}
In order to implement the hybrid scheme to the rBergomi model, we
need to introduce a particular class of non-stationary processes,
namely truncated Brownian semi-stationary ($\textit{tBss}$) processes,

\begin{equation}
\tilde{X}_{t}=\int_{0}^{t}g(t-s)\sigma_{s}dW_{s}~~t\geq0,\label{eq:tBss}
\end{equation}
where the kernel function $g(t)$, the volatility process $\sigma_{s}$
and the driving Brownian motion $W_{s}$ are as defined in the definition
of $\textit{Bss}$ processes. $\tilde{X}_{t}$ can also be seen as
the truncated stochastic integral at $0$ of the $\textit{Bss}$ process
$\bar{X}_{t}$. Equation \eqref{eq:tBss} is integrable since $g(t)$
is differentiable on $(0,\infty)$.

\subsection{Algorithm for hybrid scheme \label{subsec:Hybrid-scheme}}

Now, we can discretise equation \eqref{eq:tBss} in time. Let $N$
be the total number of time steps, $\Delta t=T/N$ be the time step
size, and $t_{0}=0\leq\ldots\leq t_{j}=j\Delta t\leq\ldots\leq t_{N}=T$
be a time grid on the interval $[0,T]$.

According to \citet{bennedsen2017hybrid}, the observations $\tilde{X}_{t_{j}}^{N},~~j=0,1,\cdots,N$
can be computed via ($\kappa=1$ case)
\begin{equation}
\tilde{X}_{t_{j}}^{N}=L_{g}\left(\Delta t\right)\sigma_{j-1}^{N}W_{j-1,1}^{N}+\sum_{k=1}^{j}g\left(b_{k}^{*}\Delta t\right)\sigma_{j-k}^{N}\bar{W}_{j-k}^{N}
\end{equation}
using the random vectors $W_{j}^{N},~~j=0,1,\cdots,N-1,$ the random
variables $\sigma_{j}^{N},~~j=0,1,\cdots,N-1,$ where $b_{k}^{*}=\left(\frac{k^{\alpha+1}-(k-1)^{\alpha+1}}{\alpha+1}\right)^{\frac{1}{\alpha}}$
, and the random vectors $\bar{W}_{i}^{N}\triangleq\int_{\frac{i}{N}}^{\frac{i+1}{N}}dW_{s}$
(see Proposition 2.8 in \citealt{bennedsen2017hybrid}).

To simulate the Volterra process $\tilde{X}$, we use:
\[
\left\{ \begin{aligned}L_{g} & \equiv1,\\
g(x) & \equiv x^{H-\frac{1}{2}},\\
\sigma(\cdot) & \equiv\sqrt{2\alpha+1}.
\end{aligned}
\right.
\]
then,
\[
\begin{aligned}W_{j-1,1}^{N} & =\int_{t_{j-1}}^{t_{j}}\left(t_{j}-s\right)^{\alpha}dW_{s}\approx\left(\frac{\Delta t}{2}\right)^{\alpha}\left(W_{t_{j}}-W_{t_{j-1}}\right)\\
\bar{W}_{j}^{N} & =\int_{t_{j}}^{t_{j+1}}dW_{s}=W_{t_{j+1}}-W_{t_{j}}\\
\sigma_{j}^{N} & =\sigma_{t_{j}}.
\end{aligned}
\]
The related matrix representation takes the form of
\begin{equation}
\left[\begin{array}{c}
\tilde{X}_{t_{1}}\\
\tilde{X}_{t_{2}}\\
\tilde{X}_{t_{3}}\\
\vdots\\
\tilde{X}_{t_{N}}
\end{array}\right]=\left[\begin{array}{ccccc}
W_{0,1} & 0 & \cdots & 0 & 0\\
W_{1,1} & g\left(b^{*}_2\Delta t\right)\bar{W}_{0} & \cdots & 0 & 0\\
W_{2,1} & g\left(b^{*}_2\Delta t\right)\bar{W}_{1} & \cdots & 0 & 0\\
\vdots & \vdots & \ddots & \vdots & \vdots\\
W_{N-1,1} & g\left(b^{*}_2\Delta t\right)\bar{W}_{N-2} & \cdots & g\left(b^{*}_{N-1}\Delta t\right)\bar{W}_{1} & g\left(b^{*}_N\Delta t\right)\bar{W}_{0}
\end{array}\right]\left[\begin{array}{c}
\text{\ensuremath{\sigma_{t_{1}}}}\\
\text{\ensuremath{\sigma_{t_{2}}}}\\
\text{\ensuremath{\sigma_{t_{3}}}}\\
\vdots\\
\text{\ensuremath{\sigma_{t_{N}}}}
\end{array}\right].\label{eq:matr}
\end{equation}
In the rBergomi model, $\sigma_{t_{i}}=\sigma$ is a constant for
$i=1,2,...,N$ defined in equation \eqref{eq:aBergomi_S1}. When simulating
$\tilde{X_{i}}$, we need to perform a matrix multiplication, the
computational complexity of which is of order $\mathcal{O}\left(N^{2}\right)$
when using the conventional matrix multiplication algorithm. However,
multiplying a lower triangular Toeplitz matrix can be regarded as
a discrete convolution which can be evaluated efficiently by fast
Fourier transform. Therefore the computational complexity can be reduced
to $\mathcal{O}\left(N\log N\right)$. The algorithm to simulate the
Volterra process $\tilde{X}$ is described in Algorithm \ref{algoX}.
Then we can use a standard Euler scheme to simulate the price $\left(S_{t_{1}},S_{t_{2}},\cdots,S_{t_{N}}\right)$.\\

\begin{algorithm2e}[H]

\DontPrintSemicolon

\SetAlgoLined

\Comment{Simulate $W_{t_{j}}$}

\While{$j=0,1,2,\cdots,N-1$}{generate random vectors $W_{t_{j}}$}

\Comment{Simulate $W_{t_{j-1},1}^{N}$}

\While{$j=1,2,\cdots,N$}{$W_{t_{j-1},1}^{N}=\left(\frac{\Delta t}{2}\right)^{\alpha}\left(W_{t_{j}}-W_{t_{j-1}}\right)$}

\Comment{Simulate $\bar{W}_{j}^{N}$}

\While{$j=0,1,2,\cdots,N-1$}{$\bar{W}_{j}^{N}=W_{t_{j+1}}-W_{t_{j}}$}
Simulate $\tilde{X}$ by the matrix multiplication \eqref{eq:matr}
using FFT

\caption{Volterra process $\tilde{X}$}

\label{algoX}\end{algorithm2e}

Below we give a simulation of the stock price in the rBergomi model
in Fig.\ref{fig:compre} (see Algorithm \ref{algoR}). The parameters
are listed in Table \ref{tab:rb}.

\begin{table}
\centering \caption{Parameters in the rBergomi model}
\begin{tabular}{|c|c|}
\hline
~~~~~~~~~~$\xi_{0}$~~~~~~~~~~  & ~~~~~~~~~0.026~~~~~~~~~ \tabularnewline
$\eta$ & 1.9 \tabularnewline
$\alpha$ & -0.43\tabularnewline
\hline
\end{tabular}\label{tab:rb}
\end{table}

\begin{algorithm2e}[H]

\DontPrintSemicolon

\SetAlgoLined

Simulate the Volterra process $\tilde{X}$ by the hybrid scheme referring
to Algorithm \ref{algoX}

\Comment{Spot variance $V_{t}$} Set initial values $V_{t}=\xi_{0}$

\While{$t=t_{1},t_{2},\cdots,t_{N}$} {$V_{t}=\xi_{0}e^{\eta\tilde{X}-\frac{\eta^{2}}{2}t^{2\alpha+1}}$}

\Comment{Log-stock price $\log(S_{t})$} Set initial values $\log(S_{t})=0$

\While{$t=t_{1},t_{2},\cdots,t_{N}$} {$\log(S_{t+\Delta t})\leftarrow\log(S_{t})+\sqrt{V_{t}}\Delta W_{t}-\frac{1}{2}V_{t}\Delta t$}

\caption{Rough Bergomi model}

\label{algoR}\end{algorithm2e}

\subsection{Approximation of the kernel}

For sake of simplicity, we start with deriving the approximation of the rBergomi model with 2 terms. It works in the same way when terms number is bigger than 2. The 2-term Bergomi model \eqref{eq:Bergomi} that we used to approximate
the rBergomi model is given as follows.
\begin{equation}
\left\{ \begin{aligned}dS_{t} & =S_{t}\sqrt{V_{t}}dW_{t}\\
d\xi_{s}^{t} & =\eta\xi_{s}^{t}\left(\alpha_{1}e^{-\kappa_{1}(t-s)}+\alpha_{2}e^{-\kappa_{2}(t-s)}\right)dB_{s},
\end{aligned}
\right.
\end{equation}
where $s\in[0,t)$. Here, we introduce the process $y_{s}^{t}$ defined
as
\begin{equation}
\left\{ \begin{aligned}y_{s}^{t} & =\alpha_{1}e^{-\kappa_{1}(t-s)}Y_{s}^{1}+\alpha_{2}e^{-\kappa_{2}(t-s)}Y_{s}^{2}\\
dY_{s}^{1} & =-\kappa_{1}Y_{s}^{1}ds+dB_{s}~~Y_{0}^{1}=0\\
dY_{s}^{2} & =-\kappa_{2}Y_{s}^{2}ds+dB_{s}~~Y_{0}^{2}=0.
\end{aligned}
\right.
\end{equation}
where $\kappa_{1},\kappa_{2}$ are from the exponential kernel $K_{\text{exp}}$,
and $Y_{s}^{1}$ and $Y_{s}^{2}$ are two O-U processes. Hence the
process $y_{s}^{t}$ can be written as a driftless Gaussian process
as follows:
\[
dy_{s}^{t}=\alpha_{1}e^{-\kappa_{1}(t-s)}dB_{s}+\alpha_{2}e^{-\kappa_{2}(t-s)}dB_{s},
\]
and its quadratic variation is $\langle dy^{t},dy^{t}\rangle_{s}=\varsigma^{2}(t-s)ds$
where $\varsigma(u)=\sqrt{\alpha_{1}^{2}e^{-2\kappa_{1}u}+\alpha_{2}^{2}e^{-2\kappa_{2}u}+2\alpha_{1}\alpha_{2}e^{-(\kappa_{1}+\kappa_{2})u}}$.
The forward variation process $\xi_{s}^{t}$ can be written as $d\xi_{s}^{t}=\eta_{s}^{t}dy_{s}^{t}$.
Thus, the solution of the forward variation process is $\xi_{s}^{t}=\xi_{0}f^{t}(s,y_{s}^{t})$
where $f^{t}(s,y)=e^{\eta y-\frac{\eta^{2}}{2}\chi(s,t)}$ and

\begin{eqnarray}
\chi(s,t) & = & \int_{t-s}^{t}\varsigma^{2}(u)du\nonumber \\
 & = & \int_{t-s}^{t}\alpha_{1}^{2}e^{-2\kappa_{1}u}+\alpha_{2}^{2}e^{-2\kappa_{2}u}+2\alpha_{1}\alpha_{2}e^{-(\kappa_{1}+\kappa_{2})u}du\nonumber \\
 & = & \alpha_{1}^{2}e^{-\kappa_{1}(t-s)}\frac{1-e^{-2\kappa_{1}s}}{2\kappa_{1}}+\alpha_{2}^{2}e^{-2\kappa_{2}(t-s)}\frac{1-e^{-2\kappa_{2}s}}{2\kappa_{2}}+2\alpha_{1}\alpha_{2}e^{-(\kappa_{1}+\kappa_{2})(t-s)}\frac{1-e^{-(\kappa_{1}+\kappa_{2})s}}{\kappa_{1}+\kappa_{2}}\label{eq:Xist}
\end{eqnarray}
Recall that $V_{t}=\xi_{t}^{t}=\xi_{0}e^{\eta y_{t}^{t}-\frac{\eta^{2}}{2}\chi(t,t)}$
and $\chi(t,t)\underset{s\rightarrow t}{\simeq}t^{2\alpha+1}$ when
$s\rightarrow t$ under the condition that the factor number is large
enough (this formula is more applicable than \eqref{eq:Xist} when
$s\rightarrow t$ , provided $n$ is large enough).

Using the approximation by Bergomi model, we consider the parameters
$\left\{ \alpha_{i},\kappa_{i}\right\} _{(i=1,2,\cdots,n)}$ in the
exponential kernel $K_{\text{exp}}=\sum_{i=1}^{n}\alpha_{i}e^{-\kappa_{i}(t-s)}$
on $s\in[0,t)$. Note that when $s\rightarrow t$, the power kernel
$K_{\text{pow}}\rightarrow\infty$ while $K_{\text{exp}}$ is finite.
To compute the approximation numerically, we need to truncate the
kernel $K_{\text{exp}}$. To do so we can use the $\texttt{scipy.optimize}$ module in $\texttt{Python}$ or the $\texttt{nlinfit}$
function in $\texttt{MATLAB}$ for the nonlinear regression of the
parameters $\left\{ \alpha_{i},\kappa_{i}\right\} _{(i=1,2,\cdots,n)}$
and the simulated price $\left\{ S_{t}\right\} $. We exemplify the
truncation of $K_{\text{exp}}$ by letting $s\in\left[0,T-\Delta t\right]$,
the truncated parameter $\theta=T-\frac{T}{N}=T-\Delta t$ and let
$T=1$.

We define the integral $I_{\text{trunc}}$ on the truncated region
$[0,\theta t)$ and apply the scaling property of Brownian motion
as follows:
\[
I_{\text{trunc}}=\sum_{i=1}^{n}\alpha_{i}\int_{0}^{\frac{\theta t}{T}}e^{-\kappa_{i}(t-s)}dB_{s}=\sum_{i=1}^{n}\alpha_{i}\sqrt{\frac{\theta}{T}}\int_{0}^{t}e^{-\kappa_{i}(1-\frac{\theta}{T})s}dB_{s}.
\]
After scaling $B_{s}$, the process $y_{s}$ demands change to be
driftless Gaussian and satisfy $y_{s}=\sum_{i=1}^{n}\alpha_{i}e^{-\kappa_{i}(1-\frac{\theta}{T})s}Y_{s}^{i}$
where $dY_{s}^{i}=\kappa_{i}(1-\frac{\theta}{T})Y_{s}^{i}ds+dB_{s}~~Y_{0}^{i}=0$.
Then the process $y_{s}$ can be written as $dy_{s}=\sum_{i=1}^{n}\alpha_{i}e^{-\kappa_{i}(1-\frac{\theta}{T})s}dB_{s}$.
Thus, the kernel in the rBergomi model on $[0,\frac{\theta}{T}t)$
can be approximated by $I_{\text{trunc}}=\sqrt{\frac{\theta}{T}}y_{t}$.

\begin{figure}[H]
\centering \includegraphics[scale=0.388]{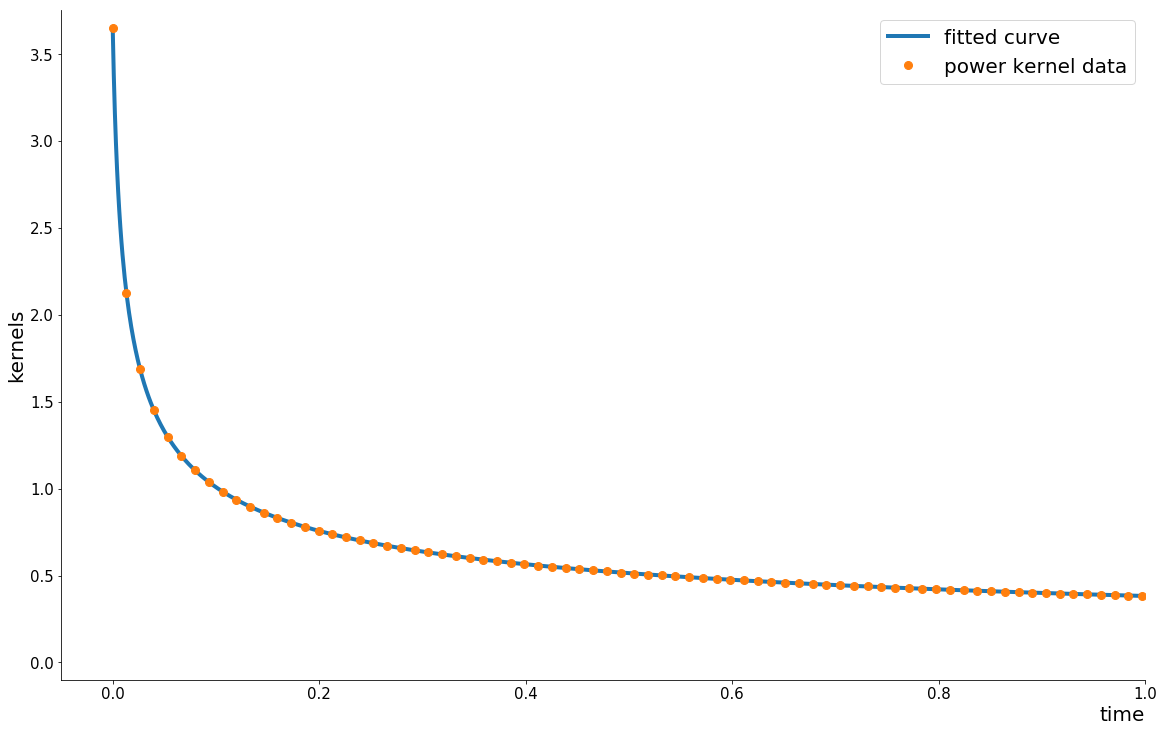} \caption{The power kernel $K_{\text{pow}}$ in the rBergomi model and the exponential
$K_{\text{exp}}$ in the 25-term aBergomi model when $T=1$ and $N=100$.}
\label{fig:appk}
\end{figure}

Figure \ref{fig:appk} displays the power kernel $K_{\text{pow}}$ in
the rBergomi model and the $K_{\text{exp}}$ in the 25-term aBergomi
model when $T=1$ and $N=100$. This figure suggests that $K_{\text{exp}}$
is sufficiently accurate for nonlinear regression, with a Root-mean-square
error (RMSE) of $1.25095\times10^{-5}$.

\begin{figure}[H]
\centering \includegraphics[scale=0.346]{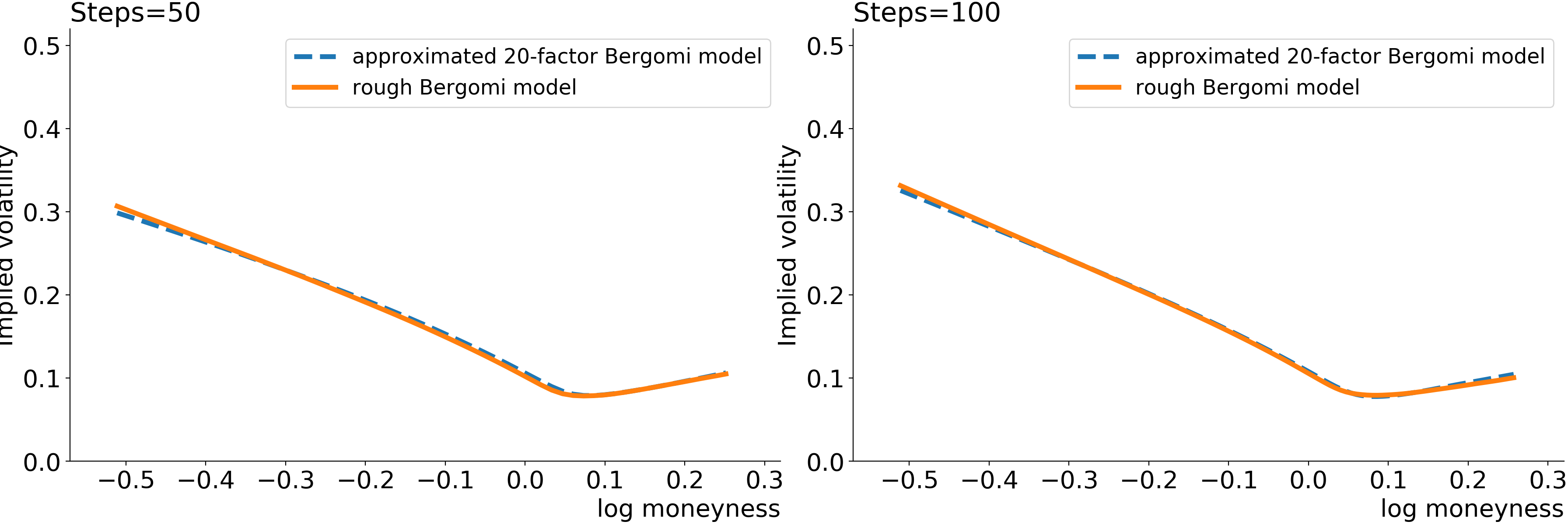} \caption{Volatility smiles for rBergomi and 25-term aBergomi models with $T=1$
simulated by $20000$ Monte Carlo paths.}
\label{fig:all_pictures}
\end{figure}

The method for simulating the variance in the aBergomi model is described
in Section \ref{subsec:Hybrid-scheme}, which leads directly to the
volatility smiles in Figure \ref{fig:all_pictures}  (see Algorithm \ref{algoB}).
From Figure \ref{fig:all_pictures}, we note that the at-the-money calibration is better with 50 time steps at the cost of a worse out-of-the-money calibration. Meanwhile, 100 time steps can approximate the rBergomi model visually well.
However, we multiply the aBergomi smile by a constant for different
time steps since the Riemann-sum scheme is able to capture the shape
of the implied volatility smile, but not its level (see \citealt{bennedsen2017hybrid}).
To generate realistic implied volatility smiles, we determine the
square of multiplication factors for different time steps in Table \ref{tab:fit}.

\begin{figure}[H]
\centering \includegraphics[scale=0.456]{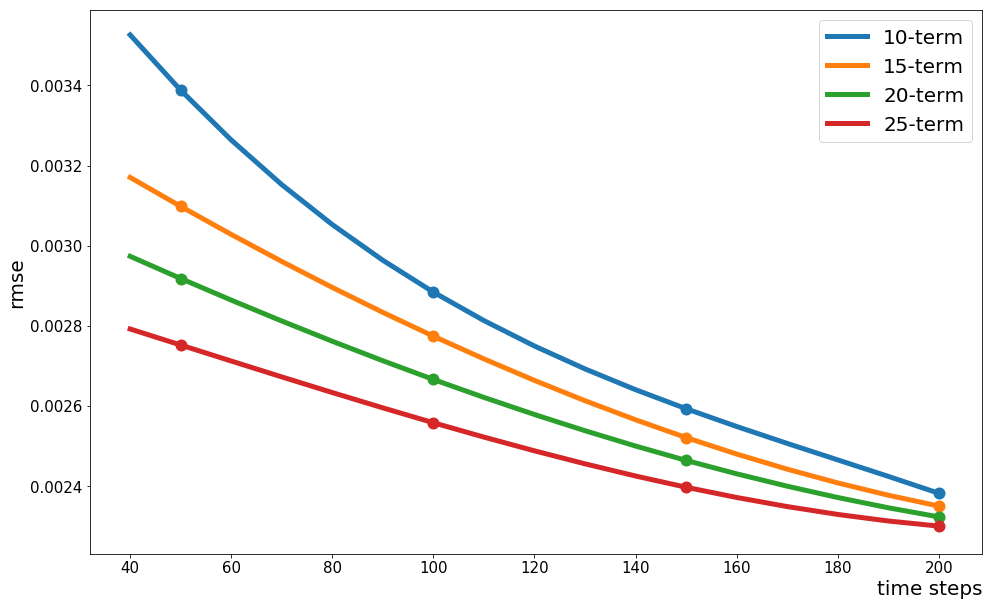} \caption{RMSE of the the implied volatility smiles of aBergomi model with three
different number of terms (15, 20 and 25) at different number of time
steps under 20000 Monte Carlo paths}
\label{fig:multi}
\end{figure}

\begin{table}
\centering \caption{The square of multiplication factors for different steps}
\begin{tabular}{|c|c|}
\hline
time steps & square of multiplication factors \tabularnewline
\hline
50  & 0.750323909 \tabularnewline
\hline
100 & 0.550447453\tabularnewline
\hline
150 & 0.485093611\tabularnewline
\hline
200 & 0.450392126\tabularnewline
\hline
\end{tabular}\label{tab:fit}
\end{table}

\begin{algorithm2e}[H]

\DontPrintSemicolon

\SetAlgoLined

\Comment{Driftless Gaussian process $y_{s}=\sum_{i=1}^{n}\alpha_{i}e^{-\kappa_{i}(1-\theta)s}Y_{s}^{i}$}
Set initial values $y_{s}=\text{zeros}(M,N)$, $Y_{0}^{i}=0$

\While{($s=t_{1},t_{2},\cdots,t_{N}$\ ) and ($i=1,2,\cdots,n$)}
{ $Y_{s+\Delta t}\leftarrow Y_{s}^{i}+\kappa_{i}(1-\theta)Y_{s}^{i}\Delta t+\Delta W_{s}$}

\Comment{Spot variance $V_{t}$} Set initial values $V_{t}=\xi_{0}$

\While{$t=t_{1},t_{2},\cdots,t_{N}$} {$V_{t}=\xi_{0}e^{\text{multiplication factor}\cdot\sqrt{\theta}y_{t}-\frac{\eta^{2}}{2}t^{2\alpha+1}}$}

\Comment{Log-stock price $\log(S_{t})$}

Set initial values $\log(S_{t})=0$

\While{$t=t_{1},t_{2},\cdots,t_{N}$} {$\log(S_{t+\Delta t})\leftarrow\log(S_{t})+\sqrt{V_{t}}\Delta W_{t}^{1}-\frac{1}{2}V_{t}\Delta t$}

\caption{n-term aBergomi model when $T=1$}

\label{algoB}\end{algorithm2e}

\begin{table}
\centering\caption{Runtime (in s) of the rBergomi model and the aBergomi model for different time steps with $T=1$ and 20000 Monte Carlo paths}
\begin{tabular}{ |c|c|c|c|c|c| } 
 \hline
 Time steps & rBergomi & 10-term aBergomi &15-term aBergomi & 20-term aBergomi & 25-term aBergomi \\ 
 \hline
 50  & 0.408105 & 0.099360 & 0.139226 & 0.172056 & 0.216408\\ 
 100 & 0.500114 & 0.236882 & 0.313049 & 0.402372 & 0.454303\\ 
 150 & 0.560219 & 0.365012 & 0.449856 & 0.558937 & 0.686911\\ 
 200 & 0.586050 & 0.425685 & 0.590825 & 0.725772 & 0.872691\\ 
 \hline
\end{tabular}\label{tab:speed}
\end{table}
We compute the RMSE of the implied volatility approximation with different
numbers of terms in the aBergomi model and different time steps in
Figure \ref{fig:multi} and compare the pricing speed in Table \ref{tab:speed}. The numerical results suggest that the RMSE
of different term numbers reduces to the same low level as the number
of time steps increases. Therefore, we may conclude that choosing
the 25-term O-U process and 100 time steps can produce a good output,
with reliable outcomes and fast calculation speed with 20000 Monte Carlo paths.

\section{Conclusion}

In this paper, we prove the power-law behavior of the ATM volatility
skew as time to maturity goes to zero of the rBergomi model and we
also propose an aBergomi model with finite terms to approximate the
rBergomi model. The approximation enables the adoption of classical
pricing methods while keeping the fractional feature of the model.
When the terms number in the aBergomi model is large enough, we can
prove its convergence to the rBergomi model. We not only give the
theoretical proofs, but also give its numerical results. A hybrid
scheme for the rBergomi model with the computational complexity $\mathcal{O}(N\log N)$
is developed for the aBergomi model. Numerically simulated results
by the hybrid scheme demonstrate the accuracy and efficiency of the
approximation.

The model parameters used in the numerical test are calibrated from
the regression of the power-law kernel of the rBergomi model. Other
efficient calibration methods are worth investigation for future research. 
\appendix
\section{Appendix}
\subsection{Proof of Theorem \ref{thm:kernel_convergence} }
In this subsection, we give the proof of the theoretical results in Theorem \ref{thm:kernel_convergence}.
\begin{proof}
Let $\left(p_{i}^{n}\right){}_{0\leq i\leq n}$ be auxiliary mean
reversion speeds such that $p_{i-1}^{n}\leq x_{i}^{n}\leq p_{i}^{n}$
for $i\leq\{1,\cdots,n\}$ and $p_{0}^{n}=0$. Recall that $K(\tau)=\int_{0}^{\infty}e^{-x\tau}\mu(dx)$.
We have
\begin{equation}
\begin{aligned}\|K^{n}-K\|_{2,T} & =\left\Vert \sum_{i=1}^{n}\alpha_{i}^{n}e^{-x_{i}^{n}\tau}-\int_{0}^{\infty}e^{-x\tau}\mu\left(dx\right)\right\Vert _{2,T}\\
 & \leq\int_{0}^{\infty}\left\Vert e^{-x(\cdot)}\right\Vert _{2,T}\mu\left(dx\right)+\sum_{i=1}^{n}\left\Vert \alpha_{i}^{n}e^{-x_{i}^{n}(\cdot)}-\int_{p_{i-1}^{n}}^{p_{i}^{n}}e^{-x(\cdot)}\mu(dx)\right\Vert _{2,T}.
\end{aligned}
\label{eq:conv1}
\end{equation}
The first term on the RHS of the inequality \eqref{eq:conv1} can
be estimated as below:
\[
\int_{p_{n}^{n}}^{\infty}\left\Vert e^{-x(\cdot)}\right\Vert _{2,T}\mu\left(dx\right)=\int_{p_{n}^{n}}^{\infty}\sqrt{\frac{1-e^{-2xT}}{2x}}\mu\left(dx\right)\leq\frac{\left(p_{n}^{n}\right){}^{-H}}{\sqrt{2}H\Gamma\left(\frac{1}{2}-H\right)}.
\]
For the second term, applying a second-order Taylor expansion of the
exponential function $e^{x}=1+x+\frac{x^{2}}{2}+\text{\ensuremath{\int_{0}^{x}\frac{(x-u)^{3}}{6}du}}$
for $t\in[0,T]$, choosing $\alpha_{i}^{n}=\int_{p_{i-1}^{n}}^{p_{i}^{n}}\mu\left(dx\right)$
and $x_{i}^{n}=\left(\frac{\int_{p_{i-1}^{n}}^{p_{i}^{n}}x^{4}\mu(dx)}{\int_{p_{i-1}^{n}}^{p_{i}^{n}}\mu(dx)}\right)^{\frac{1}{4}}$,
we can obtain that
\begin{eqnarray*}
\begin{aligned}
\left|\alpha_{i}^{n}e^{-x_{i}^{n}t}-\int_{p_{i-1}^{n}}^{p_{i}^{n}}e^{-xt}\mu\left(dx\right)\right| 
 = & \left|\alpha_{i}^{n}\left(1+(-x_{i}^{n}t)+\frac{(-x_{i}^{n}t)^{2}}{2}\right)-\int_{p_{i-1}^{n}}^{p_{i}^{n}}\left(1+\left(-xt\right)+\frac{\left(-xt\right)^{2}}{2}\right)\mu\left(dx\right)\right|\\
& +\left| \alpha_{i}^{n}\left(\int_{0}^{x_{i}^{n}t}\frac{\left(x_{i}^{n}t-u\right)^{3}}{6}du\right)-\int_{p_{i-1}^{n}}^{p_{i}^{n}}\int_{0}^{xt}\frac{\left(xt-u\right)^{3}}{6}du\mu(dx)\right| \\
  =  & \int_{p_{i-1}^{n}}^{p_{i}^{n}}\left(xt-x_{i}^{n}t\right)+\frac{\left(-x_{i}^{n}t\right)^{2}-\left(-xt\right)^{2}}{2}\mu\left(dx\right)\\
  \leq & \frac{t^{2}}{2}\int_{p_{i-1}^{n}}^{p_{i}^{n}}\left(x-x_{i}^{n}\right)^{2}\mu\left(dx\right)
 \end{aligned}
\end{eqnarray*}

since
\begin{align*}
 & \int_{p_{i-1}^{n}}^{p_{i}^{n}}\left\{ \int_{0}^{x_{i}^{n}t}\frac{(x_{i}^{n}t-u)^{3}}{6}du-\int_{0}^{xt}\frac{(xt-u)^{3}}{6}du\right\} \mu(dx)\\
= & \int_{p_{i-1}^{n}}^{p_{i}^{n}}\left\{ x_{i}^{n}t\int_{0}^{1}\frac{(x_{i}^{n}t-x_{i}^{n}ts)^{3}}{6}ds-xt\int_{0}^{1}\frac{(xt-xts)^{3}}{6}ds\right\} \mu(dx)\,\,\,,\,\,\,s=\frac{u}{xt}\\
= & \int_{p_{i-1}^{n}}^{p_{i}^{n}}\left\{ \left(x_{i}^{n}t\right)^{4}\int_{0}^{1}\frac{(1-s)^{3}}{6}ds-\left(xt\right)^{4}\int_{0}^{1}\frac{(1-s)^{3}}{6}ds\right\} \mu(dx)\\
= & \left\{ t^{4}\int_{0}^{1}\frac{(1-s)^{3}}{6}ds\right\} \int_{p_{i-1}^{n}}^{p_{i}^{n}}\left\{ \left(x_{i}^{n}\right)^{4}-\left(x\right)^{4}\right\} \mu(dx)\\
= & \left\{ t^{4}\int_{0}^{1}\frac{(1-s)^{3}}{6}ds\right\} \int_{p_{i-1}^{n}}^{p_{i}^{n}}\left\{ \left(\frac{\int_{p_{i-1}^{n}}^{p_{i}^{n}}x\mu\left(dx\right)}{\int_{p_{i-1}^{n}}^{p_{i}^{n}}\mu\left(dx\right)}\right)^{4}-\left(x\right)^{4}\right\} \mu(dx)\\
= & 0.
\end{align*}

Hence,
\[
\sum_{i=1}^{n}\left\Vert \alpha_{i}^{n}e^{-x_{i}^{n}(\cdot)}-\int_{p_{i-1}^{n}}^{p_{i}^{n}}e^{-x(\cdot)}\mu(dx)\right\Vert _{2,T}\leq\frac{T^{\frac{5}{2}}}{2\sqrt{5}}\sum_{i=1}^{n}\int_{p_{i-1}^{n}}^{p_{i}^{n}}\left(x-x_{i}^{n}\right){}^{2}\mu(dx).
\]

Thus, the convergence of $K^{n}$ depends on the weights $\alpha_{i}$
and mean reversions $x_{i}$. Let $p_{i}^{n}=i\pi_{n}$ for each $i\in\{1,\cdots,n\}$
and $\pi_{n}>0$. We have
\[
\begin{aligned}\sum_{i=1}^{n}\int_{p_{i-1}^{n}}^{p_{i}^{n}}(x-x_{i}^{n})^{2}\mu(dx) & \leq\pi_{n}^{2}\int_{0}^{p_{n}^{n}}\mu(dx)=\frac{\pi_{n}^{\frac{5}{2}-H}n^{\frac{1}{2}-H}}{\left(\frac{1}{2}-H\right)\Gamma\left(\frac{1}{2}-H\right)}\end{aligned}
\]
We can also proceed to get the explicit expressions of $\alpha_{i}^{n}$
and $x_{i}^{n}$ as follows:
\begin{eqnarray*}
\alpha_{i}^{n} & =\frac{\left(i\pi_{n}\right)^{\frac{1}{2}-H}-\left[(i-1)\pi_{n}\right]^{\frac{1}{2}-H}}{\left(\frac{1}{2}-H\right)\Gamma\left(\frac{1}{2}-H\right)},\  & x_{i}^{n}=\frac{1-2H}{3-2H}\cdot\frac{\left(i\pi_{n}\right)^{\frac{3}{2}-H}-\left[(i-1)\pi_{n}\right]^{\frac{3}{2}-H}}{\left(i\pi_{n}\right)^{\frac{1}{2}-H}-\left[(i-1)\pi_{n}\right]^{\frac{1}{2}-H}}.
\end{eqnarray*}
Since $p_{n}^{n}=n\pi_{n}\rightarrow\infty$ , we have $\pi_{n}^{\frac{5}{2}-H}n^{\frac{1}{2}-H}\rightarrow0$
as $n\rightarrow+\infty$\textcolor{blue}{{} }when $\pi_{n}<n^{-\frac{1}{6}}$
,

\begin{eqnarray}
\left\Vert K^{n}-K\right\Vert _{2,T} & \leq & \frac{1}{\sqrt{2}H\Gamma\left(\frac{1}{2}-H\right)}\left[\left(p_{n}^{n}\right)^{-H}+\frac{T^{\frac{5}{2}}H}{\sqrt{10}\left(\frac{1}{2}-H\right)}\left(p_{n}^{n}\right)^{\frac{1}{2}-H}\pi_{n}^{2}\right]\nonumber \\
 & = & \frac{1}{\sqrt{2}H\Gamma\left(\frac{1}{2}-H\right)}\left[n^{-H}\pi_{n}^{-H}+\frac{T^{\frac{5}{2}}H}{\sqrt{10}\left(\frac{1}{2}-H\right)}n^{\frac{1}{2}-H}\pi_{n}^{\frac{5}{2}-H}\right]\nonumber \\
 & = & ax^{-H}+bx^{\frac{5}{2}-H}\label{eq:conv}
\end{eqnarray}

Let $x=\pi_{n}$, RHS $y=ax^{-H}+bx^{\frac{5}{2}-H}$ and $y^{'}=-aHx^{-H-1}+b\left(\frac{5}{2}-H\right)x^{\frac{3}{2}-H}=0$,
solving for $x$, we have $x^{\frac{2}{5}}=\frac{aH}{b\left(\frac{5}{2}-H\right)}$,
where $a=n^{-H}$ and $b=\frac{T^{\frac{5}{2}}H}{\sqrt{10}\left(\frac{1}{2}-H\right)}n^{\frac{1}{2}-H}$
\begin{eqnarray*}
x & = & \pi_{n}=\left[\frac{n^{-H}H\sqrt{10}\left(\frac{1}{2}-H\right)}{T^{\frac{5}{2}}Hn^{\frac{1}{2}-H}\left(\frac{5}{2}-H\right)}\right]^{\frac{2}{5}}=\left[\frac{n^{-\frac{1}{2}}\sqrt{10}\left(\frac{1}{2}-H\right)}{T^{\frac{5}{2}}\left(\frac{5}{2}-H\right)}\right]^{\frac{2}{5}}=\frac{n^{-\frac{1}{5}}}{T}\left[\frac{\sqrt{10}\left(\frac{1}{2}-H\right)}{\left(\frac{5}{2}-H\right)}\right]^{\frac{2}{5}}
\end{eqnarray*}
When $\pi_{n}=\frac{n^{-\frac{1}{5}}}{T}\left[\frac{\sqrt{10}\left(\frac{1}{2}-H\right)}{\left(\frac{5}{2}-H\right)}\right]^{\frac{2}{5}}$,
the RHS of equation \eqref{eq:conv} attains its minimum and $\|K^{n}-K\|_{2,T}\leq Cn^{\frac{-4H}{5}}$
where $C=\frac{1}{\sqrt{2}H\Gamma\left(\frac{1}{2}-H\right)}T^{H}\left[\frac{\sqrt{10}\left(\frac{1}{2}-H\right)}{\frac{5}{2}-H}\right]^{-\frac{5}{2}H}\frac{\frac{5}{2}}{\frac{5}{2}-H}$
is a constant.
\end{proof}


\begin{thebibliography}{}
\providecommand{\natexlab}[1]{#1}
\providecommand{\url}[1]{\texttt{#1}}
\expandafter\ifx\csname urlstyle\endcsname\relax
  \providecommand{\doi}[1]{doi: #1}\else
  \providecommand{\doi}{doi: \begingroup \urlstyle{rm}\Url}\fi

\bibitem[Abi~Jaber and El~Euch, 2019]{abi2019multifactor}
Abi~Jaber, E. and El~Euch, O. (2019).
\newblock Multifactor approximation of rough volatility models.
\newblock {\em SIAM Journal on Financial Mathematics}, 10(2):309--349.

\bibitem[Bayer et~al., 2016]{bayer2016pricing}
Bayer, C., Friz, P., and Gatheral, J. (2016).
\newblock Pricing under rough volatility.
\newblock {\em Quantitative Finance}, 16(6):887--904.

\bibitem[Bayer et~al., 2018]{bayer2018hierarchical}
Bayer, C., Hammouda, C.~B., and Tempone, R. (2018).
\newblock Hierarchical adaptive sparse grids for option pricing under the rough
  {B}ergomi model.
\newblock {\em arXiv preprint arXiv:1812.08533}.

\bibitem[Bayer et~al., 2019]{bayer2019deep}
Bayer, C., Horvath, B., Muguruza, A., Stemper, B., and Tomas, M. (2019).
\newblock On deep calibration of (rough) stochastic volatility models.
\newblock {\em arXiv preprint arXiv:1908.08806}.

\bibitem[Bennedsen et~al., 2017]{bennedsen2017hybrid}
Bennedsen, M., Lunde, A., and Pakkanen, M.~S. (2017).
\newblock Hybrid scheme for {B}rownian semistationary processes.
\newblock {\em Finance and Stochastics}, 21(4):931--965.

\bibitem[Bergomi, 2005]{bergomi2005smile}
Bergomi, L. (2005).
\newblock Smile dynamics {II}.
\newblock {\em Risk magazine}, 18(10).

\bibitem[Bergomi, 2009]{bergomi2009smile}
Bergomi, L. (2009).
\newblock Smile dynamics {IV}.
\newblock {\em Risk magazine}, 22(12).

\bibitem[Bergomi and Guyon, 2012]{bergomi2012stochastic}
Bergomi, L. and Guyon, J. (2012).
\newblock Stochastic volatility's orderly smiles.
\newblock {\em Risk}, 25(5):60.

\bibitem[Carmona et~al., 2000]{carmona2000approximation}
Carmona, P., Coutin, L., and Montseny, G. (2000).
\newblock Approximation of some {G}aussian processes.
\newblock {\em Statistical inference for stochastic processes},
  3(1-2):161--171.

\bibitem[Forde and Zhang, 2017]{forde2017asymptotics}
Forde, M. and Zhang, H. (2017).
\newblock Asymptotics for rough stochastic volatility models.
\newblock {\em SIAM Journal on Financial Mathematics}, 8(1):114--145.

\bibitem[Fukasawa, 2017]{fukasawa2017short}
Fukasawa, M. (2017).
\newblock Short-time at-the-money skew and rough fractional volatility.
\newblock {\em Quantitative Finance}, 17(2):189--198.

\bibitem[Gatheral et~al., 2018]{gatheral2018volatility}
Gatheral, J., Jaisson, T., and Rosenbaum, M. (2018).
\newblock Volatility is rough.
\newblock {\em Quantitative Finance}, 18(6):933--949.

\bibitem[Gatheral and Keller-Ressel, 2019]{gatheral2019affine}
Gatheral, J. and Keller-Ressel, M. (2019).
\newblock Affine forward variance models.
\newblock {\em Finance and Stochastics}, pages 1--33.

\bibitem[Harms and Stefanovits, 2019]{harms2019affine}
Harms, P. and Stefanovits, D. (2019).
\newblock Affine representations of fractional processes with applications in
  mathematical finance.
\newblock {\em Stochastic Processes and their Applications}, 129(4):1185--1228.

\bibitem[Jacquier et~al., 2018]{jacquier2018vix}
Jacquier, A., Martini, C., and Muguruza, A. (2018).
\newblock On {VIX} futures in the rough {B}ergomi model.
\newblock {\em Quantitative Finance}, 18(1):45--61.

\bibitem[James et~al., 2013]{james2013introduction}
James, G., Witten, D., Hastie, T., and Tibshirani, R. (2013).
\newblock {\em An introduction to statistical learning}, volume 112.
\newblock Springer.

\bibitem[Jusselin and Rosenbaum, 2018]{jusselin2018no}
Jusselin, P. and Rosenbaum, M. (2018).
\newblock No-arbitrage implies power-law market impact and rough volatility.
\newblock {\em Available at SSRN 3180582}.

\bibitem[McCrickerd and Pakkanen, 2018]{mccrickerd2018turbocharging}
McCrickerd, R. and Pakkanen, M.~S. (2018).
\newblock Turbocharging {M}onte {C}arlo pricing for the rough {B}ergomi model.
\newblock {\em Quantitative Finance}, 18(11):1877--1886.

\bibitem[Muravlev, 2011]{muravlev2011representation}
Muravlev, A.~A. (2011).
\newblock Representation of a fractional {B}rownian motion in terms of an
  infinite-dimensional {O}rnstein-{U}hlenbeck process.
\newblock {\em Russian Mathematical Surveys}, 66(2):439--441.

\bibitem[{\O}ksendal and Zhang, 1993]{oksendal1993stochastic}
{\O}ksendal, B. and Zhang, T.-S. (1993).
\newblock The stochastic {V}olterra equation.
\newblock In {\em Barcelona Seminar on Stochastic Analysis}, pages 168--202.
  Springer.

\bibitem[Protter, 2005]{protter2005stochastic}
Protter, P.~E. (2005).
\newblock Stochastic differential equations.
\newblock In {\em Stochastic integration and differential equations}, pages
  249--361. Springer.

\end{thebibliography}

\end{document}